%% file: main.tex
\documentclass[a4paper,UKenglish,cleveref, autoref]{lipics-v2019}

\usepackage{svg}

\usepackage{marvosym}
\usepackage{tikz}
\usepackage{thmtools}
\usepackage{thm-restate}
\usepackage{amssymb}
\usepackage{xcolor}
\usepackage{wasysym}
\usetikzlibrary{arrows,shapes,automata, calc, chains, matrix, positioning, scopes}
\usepackage{complexity}

\title{Adaptive Synchronisation of Pushdown Automata} 

\titlerunning{Adaptive Synchronisation of Pushdown Automata}

\author{A.~R. Balasubramanian}{Technische Universit\"at M\"unchen, Munich, Germany \and \url{https://arbalan96.github.io/} }{bala.ayikudi@tum.de}{http://orcid.org/0000-0002-7258-5445}{Supported by funding from the European Research Council (ERC) under the
European Union’s Horizon 2020 research and innovation programme under grant agreement No
787367 (PaVeS).}

\author{K. S. Thejaswini}{Department of Computer Science, University of Warwick, UK}{Thejaswini.Raghavan.1@warwick.ac.uk}{}{}

\authorrunning{A.\,R. Balasubramanian and K.\,S. Thejaswini}

\Copyright{A.~R. Balasubramanian and K.~S. Thejaswini}



\category{}

\relatedversion{}

\supplement{}

\acknowledgements{We would like to thank Dmitri Chistikov for referring us to previous works on this topic}

\nolinenumbers 

\hideLIPIcs  

\EventEditors{John Q. Open and Joan R. Access}
\EventNoEds{2}
\EventLongTitle{42nd Conference on Very Important Topics (CVIT 2016)}
\EventShortTitle{CVIT 2016}
\EventAcronym{}
\EventYear{2016}
\EventDate{December 24--27, 2016}
\EventLocation{Little Whinging, United Kingdom}
\EventLogo{}
\SeriesVolume{42}
\ArticleNo{23}

\input{macros.tex}
\begin{document}

\maketitle

\begin{abstract}
We introduce the notion of adaptive synchronisation for pushdown automata,
in which there is an external observer who has no knowledge about the current
state of the pushdown automaton, but can observe the contents of the stack.
The observer would then like to decide if it is possible to bring the automaton from any state into some predetermined state by giving inputs to it in an \emph{adaptive} manner, i.e., the next input letter to be given can depend on how the contents of the stack changed after the current input letter. We show that for non-deterministic pushdown automata, this 
problem is $\TWOEXPTIME$-complete and for deterministic pushdown automata,
we show $\EXPTIME$-completeness. 

To prove the lower bounds, we first introduce (different variants of) subset-synchronisation
and show that these problems are polynomial-time equivalent with the adaptive synchronisation
problem. We then prove hardness results for the subset-synchronisation problems.
For proving the upper bounds, we consider the problem of deciding if a given alternating pushdown system has an accepting run with at most $k$ leaves and we provide an
$n^{O(k^2)}$ time algorithm for this problem.
\end{abstract}

\section{Introduction}
\input{introduction.tex}
\section{Preliminaries}
\label{section:prelim}
\input{preliminary.tex}

\section{Equivalence of Various Formulations}\label{sec:equivalence}
\input{equivalence.tex}

\section{How Hard is it to Solve Adaptive Synchronisation}\label{sec:lower-bounds}
\input{lowerbound_2EXP.tex}

\section{How Easy is it to Solve Adaptive Synchronisation}\label{sec:upperbounds}
\input{upperbound_2EXP}
\input{upperbound_EXP_det.tex}

\section{Conclusion}
\input{conclusion.tex}

\bibliography{references}

\appendix
\input{appendix.tex}

\end{document}

%% file: macros.tex
\newcommand{\Pc}{\mathcal{P}}

\newcommand{\Ac}{\mathcal{A}}
\newcommand{\Mc}{\mathcal{M}}
\newcommand{\Cc}{\mathcal{C}}

\newcommand{\Nc}{\mathcal{N}}
\newcommand{\init}{\mathit{init}}
\newcommand{\fin}{\mathit{fin}}

\newcommand{\acc}{\mathit{acc}}
\newcommand{\rej}{\mathit{rej}}

\newcommand{\test}{\mathtt{test}}
\newcommand{\cmmd}{\mathtt{cmd}}

\newcommand{\inputt}{\mathtt{in}}

\newcommand{\set}[1]{[#1]}

\newcommand{\Succ}{\mathit{Succ}}
\newcommand{\End}{\mathtt{end}}

\newcommand{\zero}{\mathbf{0}}

\newcommand{\bDiamond}{\mathbin{\Diamond}}

\newcommand{\mycirc}[1][black]{\Large\textcolor{#1}{\ensuremath\bullet}}

\usepackage{mathtools}
\usepackage{mathrsfs}
\usepackage{eurosym}

\newcommand{\xhri}{\xhookrightarrow}

\newcommand{\word}{\mathit{word}}
\newcommand{\nex}{\mathit{next}}

\usepackage{todonotes}

%% file: introduction.tex
The notion of a synchronizing word for finite-state machines
is a classical concept in computer science
which consists of deciding, given a finite-state machine,
whether there is a word which brings all of its states to a
single state. Intuitively, assuming that we initially do not know
which state the machine is in, such a word \emph{synchronises}
it to a single state and assists in regaining control over the machine.
This idea has been studied for many types of finite-state machines~\cite{Cerny, Survey, Unambiguous, MDPs} with applications in biocomputing~\cite{DNA}, 
planning and robotics~\cite{Eppstein, Natarajan} and
testing of reactive systems~\cite{Testing,Hennie}.
In recent years, the notion of a synchronizing word 
has been extended to various \emph{infinite-state systems}
such as timed automata \cite{Timed}, register automata \cite{Register}, nested word automata \cite{Nested}, pushdown and visibly pushdown automata \cite{VPA,PushdownAutomataUndecidable}.
In particular, for the pushdown case, Fernau, Wolf and Yamakami \cite{PushdownAutomataUndecidable} have shown
that this problem is undecidable even for deterministic pushdown automata.

When the finite-state machine can produce outputs, the notion of synchronisation 
could be further refined to give rise to synchronisation under \emph{partial observation} or \emph{adaptive synchronisation} (See Chapter 1 of~\cite{Lecturenotes} and~\cite{PartialObservation2}). In this setting, there is an external observer
who does not know the current state of the machine, however she can give inputs
to the machine and observe the outputs given by the machine. Depending
on the outputs of the machine, she can \emph{adaptively} decide which input letter
to give next. In this manner, the observer would like to bring the 
machine into some predetermined state. Larsen, Laursen and Srba~\cite{PartialObservation2} describe an example of adaptive synchronisation pertaining to the orientation of  a simplified model of satellites,
in which they observe that adaptively choosing the input letter is sometimes necessary in order to achieve synchronisation.
In this paper, we extend this notion of adaptive synchronisation to pushdown automata (PDA). In our model, the observer does not know which state the PDA is currently in, but can observe the contents of the stack.
She would then like to decide if it is possible to synchronise the PDA into some state by giving inputs to the PDA
adaptively, i.e., depending on how the stack changes after each input.
To the best of our knowledge, the notion of adaptive synchronisation has not been considered before for any class of infinite-state systems.

This question is a natural extension of the notion of adaptive synchronisation from finite-state machines to pushdown automata.
Further, it is mentioned in the works of Lakhotia, Uday Kumar and Venable as well as Song and Touili \cite{Malware1,Malware2}
that several antivirus systems determine whether a program is malicious by observing the 
calls that the program makes to the operating system. With this in mind,
Song and Touili use pushdown automata~\cite{Malware1} as abstractions of programs where a stack stores the calls made by the program and use this abstraction to detect viruses. Hence, we believe that our setting of being able to observe the changes happening to the stack can be practically motivated.

Our main results regarding adaptive synchronisation are as follows: We show that for non-deterministic pushdown automata, the problem is $\TWOEXPTIME$-complete. However, by restricting our input to deterministic pushdown automata, we show that we can get $\EXPTIME$-completeness, thereby obtaining an exponential reduction in complexity. 

We also consider a natural variant of this problem, called \emph{subset adaptive synchronisation}, which is similar to adaptive synchronisation, except the observer has more knowledge about which state the automaton is initially in. We obtain a surprising result that shows that this variant is polynomial-time equivalent to adaptive synchronisation, unlike in the case of finite-state machines. Furthermore, for the deterministic case of this variant, we obtain an algorithm that runs in time $O\left(n^{ck^3}\right)$ where $n$ is the size of the input and $k$ is the size of the subset of states that the observer believes the automaton is initially in. This gives a polynomial time algorithm if $k$ is fixed and a quasi-polynomial time algorithm if $k = O(\log n)$.

Used as a subroutine in the above decision procedure, is an $O\left(n^{ck^2}\right)$ time algorithm to the following question, which we call the \emph{sparse-emptiness problem}: Given an alternating pushdown system and a number $k$, decide whether there is an accepting run of the system with at most $k$ leaves. Intuitively, such a run means that the system has an accepting run in which it uses only ``limited universal branching''. We note that such a notion of alternation 
with ``limited universal branching'' has recently been studied by Keeler and Salomaa for alternating finite-state automata~\cite{LimitedUniversalBranching}. 
Our problem can be considered as a generalisation of one of their problems (Corollary 2 of~\cite{LimitedUniversalBranching}) to pushdown systems.  We think that this problem and its associated algorithm might be of independent interest. 

\textbf{Roadmap:} In Section~\ref{section:prelim}, we introduce notations. In Section~\ref{sec:equivalence}, we discuss different variations of the problem. In Sections~\ref{sec:lower-bounds} and~\ref{sec:upperbounds} we prove lower and upper bounds respectively.
Due to lack of space, some of the proofs can be found in the appendix.

%% file: preliminary.tex
Given a finite set $X$, we let $X^*$ denote the set of all words with the alphabet $X$.
As usual, the concatenation of two words $x,y \in X^*$ is denoted by $xy$.

\subsection{Pushdown Automata} 
We recall the well-known notion of a pushdown automaton.
A pushdown automaton (PDA) is a 4-tuple $\Pc = (Q,\Sigma,\Gamma, \delta)$ where $Q$ is a finite set of \emph{states}, $\Sigma$ is the \emph{input alphabet}, $\Gamma$ is the \emph{stack alphabet} 
and $\delta \subseteq (Q\times\Sigma\times\Gamma)\times (Q\times \Gamma^*)$ is the \emph{transition relation}. Alternatively, sometimes we will describe the transition relation
$\delta$ as a function $Q \times \Sigma \times \Gamma \mapsto 2^{Q \times \Gamma^*}$.
We will always use small letters $a,b,c,\dots$ to denote elements 
of $\Sigma$,
capital letters $A,B,C,\dots$ to denote elements of $\Gamma$ and Greek letters $\gamma,\eta,\omega,\dots$ to denote
elements of $\Gamma^*$. 

If $(p,a,A,q,\gamma) \in \delta$ then 
we sometimes denote it by $(p,A) \xhri{a} (q,\gamma)$. We say $A$ is  
the \emph{top} of the stack that is \emph{popped} and $\gamma$ is the string that is \emph{pushed}
onto the stack. A \emph{configuration} of the automaton is a tuple $(q,\gamma)$ where $q \in Q$ and $\gamma \in \Gamma^*$. Given two configurations $(q,A\gamma)$ and $(q',\gamma' \gamma)$ of $\Pc$ with $A \in \Gamma$, we say that $(q,A \gamma) \xrightarrow{a} (q', \gamma' \gamma)$ iff 
$(q,A) \xhri{a} (q',\gamma')$. 

As is usual, we assume that there exists a special 
\emph{bottom-of-the-stack} symbol $\bot\in \Gamma$,
such that whenever some transition pops $\bot$, it pushes
it back in the bottom-most position.
A PDA is said to be deterministic if for every $q \in Q$, $a \in \Sigma$ and $A \in \Gamma$, $\delta(q,a,A)$ has exactly one element. If a PDA is deterministic, 
we further abuse notation and denote $\delta(q,a,A)$ as a single element and not as a set.

\subsection{Adaptive Synchronisation}\label{subsec:define-ada-sync}
We first expand upon the intuition given in the introduction for adaptive synchronisation with the help of a running example. Consider the pushdown automaton as given in Figure~\ref{fig:automata} where we do not know which state the automaton is in currently, but we do know that the stack content is $\bot$.  
To synchronise the automaton to the state $4$ when the stack is visible, the observer has a strategy as depicted in Figure~\ref{fig:tree}. The labelling of the nodes of the tree intuitively denotes the `knowledge of the observer' at the current point in the strategy and the labelling of the edges denotes
the letter that she inputs to the PDA. Initially, according to the observer, the automaton could be in any one of the 4 states. The observer first inputs the letter $\Box$. If the top of the stack becomes $\mycirc[red]$, then she knows that the automaton is currently either in state $1$ or $2$.
On the other hand, if the top of the stack becomes $\mycirc[blue]$, then the observer can deduce that the automaton is currently in state $3$ or $4$.
From these two scenarios, by following the appropriate strategy depicted in the figure,
we can see that she can synchronise the automaton to state 4. However, if the stack was hidden to the observer, reading either $\bDiamond$ or $\Box$ does not change the knowledge of the observer and therefore, there is no word that can be read that would synchronise the automaton to any state.
\begin{figure}
\begin{minipage}{0.45 \textwidth}
    \centering
    \begin{tikzpicture}[shorten >=1pt,node distance=2cm,auto]
  \tikzstyle{every state}=[fill={rgb:black,1;white,10}]
  \node[state] (q_1)   {$1$};
  \node[state]           (q_2) [below =2.5cm of q_1]     {$2$};
  \node[state,accepting] (r_1) [right=2.5cm of q_1] {$4$};
  \node[state]           (r_2) [below=2.5cm of r_1]     {$3$};
  \path[->]
  (q_1) edge [loop above]  node {$\Box,pop(\mycirc[blue]/\mycirc[red])$} (   )
        edge [bend left=10]  node[align=center] {$\bDiamond$\\$\mycirc[red]\rightarrow\mycirc[blue]\mycirc[red]$\\$\mycirc[blue]\rightarrow\mycirc[blue]\mycirc[blue]$} (q_2)
        edge [bend left=10]  node {$\Box,\bot\to\mycirc[blue]\bot$} (r_1)
  (q_2) edge [loop below]  node {$\Box,pop(\mycirc[blue]/\mycirc[red])$} (   )
        edge [bend left=10]  node[align=center] {$\bDiamond$\\$\mycirc[red]\rightarrow\mycirc[red]\mycirc[red]$\\$\mycirc[blue]\rightarrow\mycirc[red]\mycirc[blue]$} (q_1)
        edge [bend left=10]  node {$\Box,\bot\rightarrow\mycirc[blue]\bot$} (r_2)
  (r_1) edge [loop above] node {$\Box,pop(\mycirc[blue]/\mycirc[red])$} (   )
        edge [bend left=10]  node[align=center] {$\bDiamond$\\$\mycirc[red]\rightarrow\mycirc[red]\mycirc[red]$\\$\mycirc[blue]\rightarrow\mycirc[red]\mycirc[blue]$} (r_2)
        edge [bend left=10]  node {$\Box,\bot\rightarrow\mycirc[red]\bot$} (q_1)
  (r_2) edge [loop below] node {$\Box,pop(\mycirc[blue]/\mycirc[red])$} (   )
        edge [bend left=10]  node[align=center] {$\bDiamond$ \\$\mycirc[red]\rightarrow\mycirc[blue]\mycirc[red]$\\$\mycirc[blue]\rightarrow\mycirc[blue]\mycirc[blue]$} (r_1)
        edge [bend left=10]  node {$\Box,\bot\rightarrow\mycirc[red]\bot$} (q_2);
\end{tikzpicture}
    \caption{A label of $a,A\rightarrow\gamma$ means that if the input is $a$ and if the top of the stack is $A$, then pop $A$ and push $\gamma$.\label{fig:automata}} 
\end{minipage}
\hfill \vline \hfill
\begin{minipage}{0.45 \textwidth}
    \centering
\makebox[\textwidth][c]{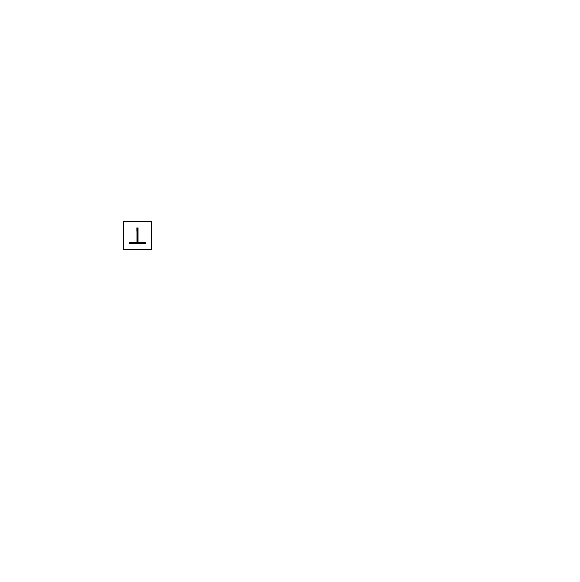}%
\caption{A synchroniser between $(\{1,2,3,4\},\bot)$ and state $4$ for the PDA in Figure~\ref{fig:automata}.\label{fig:tree}}
\end{minipage}
\end{figure}

We now formalize the notion of an adaptive synchronizing word that we have so far described.
Let $\Pc = (Q,\Sigma, \Gamma, \delta)$ be a PDA.
Given $S \subseteq Q$, $a \in \Sigma$ and $A \in \Gamma$,
let $T_{S,A}^{a} := \{t \in \delta\mid t = (p,a,A, q, \gamma)$ where $p\in S\}$. 
Intuitively, if the observer knows that $\Pc$ is currently in some state in $S$ and the top of the stack
is $A$ and she chooses to input $a$, then $T_{S,A}^{a}$ is the set of transitions that might take place.
We define an equivalence relation $\sim_{S,A}^{a}$ on the elements of  $T_{S,A}^{a}$ as follows:
$t_1 \sim_{S,A}^a t_2 \iff \exists \gamma \in \Gamma^*$ such that $t_1 = (p_1, a, A, q_1,\gamma)$ and $ t_2 = (p_2, a, A, q_2,\gamma)$. Notice that if $t_1 \sim_{S,A}^a t_2$ then the observer cannot distinguish occurrences of $t_1$ from occurrences of $t_2$.
In our running example, if we take $S = \{3,4\}$, $a = \bDiamond$ and $A = \mycirc[blue]$,
it is easy to see that $T^a_{S,A}$ is $\{(3,\bDiamond,\mycirc[blue],4,\mycirc[blue]\mycirc[blue]),
(4,\bDiamond,\mycirc[blue],3,\mycirc[red]\mycirc[blue])\}$ and these two transitions 
are not in the same equivalence class under $\sim^a_{S,A}$.

The relation $\sim_{S,A}^a$ partitions the elements of $T_{S,A}^{a}$ into equivalence classes. If $E$ is an equivalence class of $\sim_{S,A}^a$, then notice that there is a word $\gamma \in \Gamma^*$ such that all the transitions in $E$ pop $A$ and push $\gamma$ onto the stack. This word $\gamma$ will 
be denoted by $\word(E)$.
If we define $\nex(E) := \{q \mid  (p,a,A, q, \word(E)) \in E\}$, then
$\nex(E)$ contains all the states that the automaton can move to
if \emph{any} of the transitions from $E$ occur.
Now, suppose the observer knows that $\Pc$ is currently in some state in $S$ with $A$ being at the top of the stack. Assuming she inputs the letter $a$ and observes that $A$ has been popped and $\word(E)$ has been pushed, she can deduce that $\Pc$ is currently in some state in $\nex(E)$.
In our running example of $S = \{3,4\}$, $a = \bDiamond$ and $A = \mycirc[blue]$,
there are two equivalence classes $E_1 = \{(3,\bDiamond,\mycirc[blue],4,\mycirc[blue]\mycirc[blue])\}$ and $E_2 = \{(4,\bDiamond,\mycirc[blue],3,\mycirc[red]\mycirc[blue])\}$ 
with $\nex(E_1) = \{3\}$, $\nex(E_2) = \{4\}$, 
$\word(E_1) = \{\mycirc[blue]\mycirc[blue]\}$ and $\word(E_2) = \{\mycirc[red]\mycirc[blue]\}$.

A \emph{pseudo-configuration} of the automaton $\Pc$ is a pair $(S,\gamma)$ 
such that $S \subseteq Q$ and $\gamma \in \Gamma^*$. The pseudo-configuration $(S,\gamma)$
captures the knowledge of the observer at any given point. Given a pseudo-configuration $(S,A\gamma)$
and an input letter $a$, let $\Succ(S,A\gamma,a) := \{(\nex(E_1),\word(E_1)\gamma),\dots,
(\nex(E_k),\word(E_k)\gamma)\}$ where $E_1,\dots,E_k$ are the equivalence
classes of $\sim^a_{S,A}$.  Each element of $\Succ(S,A\gamma,a)$ will be 
called a possible successor of $(S,A\gamma)$ under the input letter $a$.
The function $\Succ$ captures all the possible
pseudo-configurations that could happen when the observer inputs $a$ at the pseudo-configuration $(S,A\gamma)$. 

We now define the notion of a \emph{synchroniser} which will correspond to a strategy 
for the observer to synchronise the automaton into some state. 
Let $I \subseteq Q, s\in Q$ and $\gamma \in \Gamma^*$. (The $I$ stands for \emph{\textbf{I}nitial set of states},
and the $s$ stands for \emph{\textbf{s}ynchronising state}).
A \emph{synchroniser} between the pseudo-configuration $(I,\gamma)$ and the state $s$, is a labelled tree $T$
such that
\begin{itemize}
    \item All the edges are labelled by some input letter $a \in \Sigma$ such that, for every vertex $v$, all its outgoing edges have the same label. 
    \item The root is labelled by the pseudo-configuration $(I,\gamma)$. 
    \item Suppose $v$ is a vertex which is labelled by the pseudo-configuration $(S,A\eta)$.
    Let $a$ be the unique label of its outgoing edges and let $\Succ(S,A\eta,a)$ be of size $k$.
    Then $v$ has $k$ children, with the $i^{th}$ child labelled by the $i^{th}$ pseudo-configuration in $\Succ(S,A\eta,a)$.
    \item For every leaf, there exists $\eta \in \Gamma^*$ such that its label is $(\{s\},\eta)$.
\end{itemize}
In addition, if all the leaves are labelled by $(\{s\},\bot)$, then $T$ is called
a \emph{super-synchroniser} between $(I,\gamma)$ and $s$. 
We use the notation $(I,\gamma) \xRightarrow[\Pc]{} s$ (resp. $(I,\gamma) \xRightarrow[\Pc]{\text{sup}} s$) to denote that there is a synchroniser (resp. super-synchroniser) between $(I,\gamma)$
and $s$ in the PDA $\Pc$. (When $\Pc$ is clear from context, we would drop
it from the arrow notation).

\subsection{Different Formulations}
We now formally introduce the problem which we will refer to as \emph{adaptive synchronising problem} ({\scshape Ada-Sync}) and it is defined as the following:
\begin{quote}
    \noindent \emph{Given: } A PDA $\Pc = (Q,\Sigma,\Gamma,\delta)$ and a word $\gamma \in \Gamma^*$\\
    \noindent \emph{Decide: } Whether there is a state $s$ such that $(Q,\gamma) \Rightarrow s$
\end{quote}

The {\scshape Det-Ada-Sync} problem is the same as {\scshape Ada-Sync}, except that
the given pushdown automaton is deterministic. Notice that we can generalise the adaptive synchronising problem by the following \emph{subset adaptive synchronising problem} ({\scshape Subset-Ada-Sync}): Given a PDA $\Pc = (Q,\Sigma,\Gamma,\delta)$, a subset $I \subseteq Q$ and a word $\gamma \in \Gamma^*$, decide if there is a state $s$
such that $(I,\gamma) \Rightarrow s$. Similarly, we can define {\scshape Det-Subset-Ada-Sync}.

\begin{remark}
    One can also frame both of these problems in various other
    ways such as ``Given $\Pc,\gamma$ and $q$ does $(Q,\gamma) \Rightarrow q$?''
    or ``Given $\Pc,\gamma,I$, is there a $q$ such that $(I,\gamma) \xRightarrow{\text{sup}} q$'' etc. 
    We chose this version, because this is similar to the way it is defined for the finite-state version (Problem 1 of ~\cite{PartialObservation2}). Nevertheless, in order to make the lower bounds easier
    to understand, we introduce a few different variants of {\scshape Ada-Sync} and {\scshape Subset-Ada-Sync} in Section~\ref{sec:equivalence} and
    conclude that
    that they are all polynomial-time equivalent with {\scshape Ada-Sync}. 
    We defer a detailed analysis of the different variants of this problem to future work.
\end{remark}
\begin{remark}
    One can relax the notion of a synchroniser and ask instead for an adaptive ``homing'' word, which is the same as a synchroniser, except that we now only require that if $(S,\gamma)$ is the label of a leaf then $S$ is \emph{any} singleton. 
    Intuitively, in an adaptive homing word, we are content with knowing the state the automaton is in after
    applying the strategy, rather than enforcing the automaton to synchronise into some state.
    Due to lack of space, we state this problem formally and  
    prove in the appendix that it is polynomial-time equivalent
    to {\scshape Ada-Sync}. In the main paper, we primarily focus on finding the complexity status of the problems {\scshape Ada-Sync} and {\scshape Subset-Ada-Sync}.
\end{remark}

The main results of this paper are now as follows: 
\begin{theorem}
{\scshape Ada-Sync} and {\scshape Subset-Ada-Sync} are both $\TWOEXPTIME$-complete. {\scshape Det-Ada-Sync} and {\scshape Det-Subset-Ada-Sync} are both $\EXPTIME$-complete.
\end{theorem}


%% file: tree.pdf_tex
\begingroup%
  \makeatletter%
  \providecommand\color[2][]{%
    \errmessage{(Inkscape) Color is used for the text in Inkscape, but the package 'color.sty' is not loaded}%
    \renewcommand\color[2][]{}%
  }%
  \providecommand\transparent[1]{%
    \errmessage{(Inkscape) Transparency is used (non-zero) for the text in Inkscape, but the package 'transparent.sty' is not loaded}%
    \renewcommand\transparent[1]{}%
  }%
  \providecommand\rotatebox[2]{#2}%
  \newcommand*\fsize{\dimexpr\f@size pt\relax}%
  \newcommand*\lineheight[1]{\fontsize{\fsize}{#1\fsize}\selectfont}%
  \ifx\svgwidth\undefined%
    \setlength{\unitlength}{166.99464414bp}%
    \ifx\svgscale\undefined%
      \relax%
    \else%
      \setlength{\unitlength}{\unitlength * \real{\svgscale}}%
    \fi%
  \else%
    \setlength{\unitlength}{\svgwidth}%
  \fi%
  \global\let\svgwidth\undefined%
  \global\let\svgscale\undefined%
  \makeatother%
  \begin{picture}(1,1.00241813)%
    \lineheight{1}%
    \setlength\tabcolsep{0pt}%
    \put(0,0){\includegraphics[width=\unitlength,page=1]{tree.pdf}}%
    \put(0.03897967,0.57938833){\color[rgb]{0,0,0}\makebox(0,0)[lt]{\lineheight{1.25}\smash{\begin{tabular}[t]{l}$\{1,2\}$\end{tabular}}}}%
    \put(0.45875842,0.58926399){\color[rgb]{0,0,0}\makebox(0,0)[lt]{\lineheight{1.25}\smash{\begin{tabular}[t]{l}$\{4\}$\end{tabular}}}}%
    \put(0.6957134,0.42835521){\color[rgb]{0,0,0}\makebox(0,0)[lt]{\lineheight{1.25}\smash{\begin{tabular}[t]{l}$\{4\}$\end{tabular}}}}%
    \put(0.02991485,0.25240741){\color[rgb]{0,0,0}\makebox(0,0)[lt]{\lineheight{1.25}\smash{\begin{tabular}[t]{l}$\{4\}$\end{tabular}}}}%
    \put(0.26276805,0.25761461){\color[rgb]{0,0,0}\makebox(0,0)[lt]{\lineheight{1.25}\smash{\begin{tabular}[t]{l}$\{3\}$\end{tabular}}}}%
    \put(0.27064195,0.08299832){\color[rgb]{0,0,0}\makebox(0,0)[lt]{\lineheight{1.25}\smash{\begin{tabular}[t]{l}$\{4\}$\end{tabular}}}}%
    \put(0,0){\includegraphics[width=\unitlength,page=2]{tree.pdf}}%
    \put(0.25117451,0.948333){\color[rgb]{0,0,0}\makebox(0,0)[lt]{\lineheight{1.25}\smash{\begin{tabular}[t]{l}$\{1,2,3,4\}$\end{tabular}}}}%
    \put(0.48225521,0.77415183){\color[rgb]{0,0,0}\makebox(0,0)[lt]{\lineheight{1.25}\smash{\begin{tabular}[t]{l}$\{3,4\}$\end{tabular}}}}%
    \put(0.04165181,0.41537538){\color[rgb]{0,0,0}\makebox(0,0)[lt]{\lineheight{1.25}\smash{\begin{tabular}[t]{l}$\{3,4\}$\end{tabular}}}}%
    \put(0.05172824,0.76210519){\color[rgb]{0,0,0}\makebox(0,0)[lt]{\lineheight{1.25}\smash{\begin{tabular}[t]{l}$\{1,2\}$\end{tabular}}}}%
    \put(0.69341022,0.58783781){\color[rgb]{0,0,0}\makebox(0,0)[lt]{\lineheight{1.25}\smash{\begin{tabular}[t]{l}$\{3\}$\end{tabular}}}}%
    \put(0.36446568,0.87673369){\color[rgb]{0,0,0}\makebox(0,0)[lt]{\lineheight{1.25}\smash{\begin{tabular}[t]{l}$\Box$\end{tabular}}}}%
    \put(0.13398342,0.6842532){\color[rgb]{0,0,0}\makebox(0,0)[lt]{\lineheight{1.25}\smash{\begin{tabular}[t]{l}$\Box$\end{tabular}}}}%
    \put(0.13852017,0.49824611){\color[rgb]{0,0,0}\makebox(0,0)[lt]{\lineheight{1.25}\smash{\begin{tabular}[t]{l}$\Box$\end{tabular}}}}%
    \put(0.58661631,0.66846571){\color[rgb]{0,0,0}\makebox(0,0)[lt]{\lineheight{1.25}\smash{\begin{tabular}[t]{l}$\bDiamond$\end{tabular}}}}%
    \put(0.76005618,0.50731961){\color[rgb]{0,0,0}\makebox(0,0)[lt]{\lineheight{1.25}\smash{\begin{tabular}[t]{l}$\bDiamond$\end{tabular}}}}%
    \put(0.16111396,0.33047719){\color[rgb]{0,0,0}\makebox(0,0)[lt]{\lineheight{1.25}\smash{\begin{tabular}[t]{l}$\bDiamond$\end{tabular}}}}%
    \put(0.34022482,0.17727046){\color[rgb]{0,0,0}\makebox(0,0)[lt]{\lineheight{1.25}\smash{\begin{tabular}[t]{l}$\bDiamond$\end{tabular}}}}%
    \put(0,0){\includegraphics[width=\unitlength,page=3]{tree.pdf}}%
  \end{picture}%
\endgroup%

%% file: equivalence.tex
In this section, we show that the problems {\scshape Ada-Sync} and {\scshape Subset-Ada-Sync} are
polynomial-time equivalent to each other. A similar result is also shown for their corresponding deterministic versions.
We note that such a result is not true for finite-state (Moore) machines (Table 1 of~\cite{PartialObservation2}) and so we provide a proof of this here, because
it illustrates the significance of the stack in the pushdown version.
\begin{lemma}\label{lem:equiv-subset-sync}
    {\scshape Ada-Sync} (resp. {\scshape Det-Ada-Sync}) is polynomial time
    equivalent to {\scshape Subset-Ada-Sync} (resp. {\scshape Det-Subset-Ada-Sync}).
\end{lemma}

\begin{proof}
    It suffices to show that
    {\scshape Subset-Ada-Sync} (resp. {\scshape Det-Subset-Ada-Sync})
    can be reduced to  {\scshape Ada-Sync} (resp. {\scshape Det-Ada-Sync})  
    in polynomial time.
    
    Let $\Pc = (Q,\Sigma,\Gamma,\delta)$ be a PDA with $I \subseteq Q$
    and $\gamma \in \Gamma^*$. Let $q_{I}$ be some fixed state in the subset $I$.
    Construct $\Pc'$ from $\Pc$ by  
    adding a new stack letter $\#$ and the following new transitions: Upon reading any $a \in \Sigma$, 
    if the top of the stack is $\#$, then any state $q \in I$ pops $\#$ and stays at $q$ whereas any state $q \notin I$ pops $\#$ and moves to $q_{I}$. 
    Notice that $\Pc'$ is deterministic if $\Pc$ is.
    
    It is clear that if $(I,\gamma) \xRightarrow[\Pc]{} s$ for some state $s$, then $(Q,\#\gamma) \xRightarrow[\Pc']{} s$. We now claim that the
    other direction is true as well. To see this, suppose there
    is a synchroniser in $\Pc'$ (say $T$) between $(Q,\#\gamma)$ and some state $s$. It is easy to see that, irrespective of the label of 
    the outgoing edge from the root of $T$, there is only one child of the root 
    which is labelled by $(I,\gamma)$. Now, no transition pushes $\#$
    onto the stack and so nowhere else in the synchroniser does $\#$ 
    appear in the label of some vertex. It is then
    easy to see that if we remove the root of $T$, we 
    get a synchroniser between $(I,\gamma)$ and $s$ in $\Pc$.
\end{proof}

Lemma~\ref{lem:equiv-subset-sync} allows us to 
introduce a series of problems which we can prove are poly-time equivalent
to {\scshape Ada-Sync}. The reason to consider these problems is that lower bounds for these are substantially easier to prove than for {\scshape Ada-Sync}. The three problems are 
as follows:
\begin{enumerate}
    \item {\scshape Given-Sync}: Given a PDA $\Pc$, 
a subset $I$, a word $\gamma$ and also \emph{a state $s$}, check if  $(I,\gamma) \Rightarrow s$.
    \item {\scshape Super-Sync} has the same input as {\scshape Given-Sync},
except we ask if  $(I,\gamma) \xRightarrow{\text{sup}} s$.
    \item {\scshape Special-Sync} is the same as {\scshape Super-Sync} but restricted to inputs where $\gamma$ is $\bot$.
\end{enumerate}

\begin{restatable}{lemma}{ThmEquivSuperSync}\label{lem:equiv-super-sync}
    {\scshape Subset-Ada-Sync}, {\scshape Given-Sync}, {\scshape Super-Sync}
    and {\scshape Special-Sync} are all poly. time equivalent.
    Further the same applies for their corresponding deterministic versions.
\end{restatable}

Because of this lemma, for the rest of this paper, we will only be concerned with the {\scshape Special-Sync} problem, where given a PDA $\Pc$, 
a subset $I$ and a state $s$, we have to decide if $(I,\bot) \xRightarrow[]{\text{sup}} s$.

%% file: lowerbound_2EXP.tex
To prove the lower bounds, we introduce the notion of an alternating extended pushdown system (AEPS), which is an extension of pushdown systems with Boolean variables and alternation. 
\subsection{Alternating Extended Pushdown Systems}

An \emph{alternating extended pushdown system} (AEPS) $\Ac$ is a tuple $(Q,V,\Gamma,\Delta,\init,\fin)$ where
$Q$ and $V$ are finite sets of states and Boolean variables respectively, $\Gamma$ is the stack
alphabet, $\init, \fin \in Q$ are the initial and final states respectively.
$\Ac$ has no input letters but it has a stack to which it can pop and push letters from $\Gamma$.
Each variable in $V$ is of Boolean type and a transition of $\Ac$ could apply simple tests on these variables and depending on the outcome, can update their values.
A configuration of $\Ac$ is a tuple $(q,\gamma,F)$ where $q \in Q, \gamma \in \Gamma^*$
and $F : V \to \{0,1\}$ is a function assigning a Boolean value to each variable.

Let $\test$ denote the set of \emph{tests} given by $\{v ?= b : v \in V, b \in \{0,1\}\}$ and
let $\cmmd$ denote the set of \emph{commands} given by $\{v \mapsto b : v \in V, b \in \{0,1\} \}$.
A \emph{consistent command} is a conjunction of elements
from $\cmmd$ such that for every $v \in V$, both $v \mapsto 0$ and $v \mapsto 1$
are not present in $\cmmd$.  
The transition relation $\Delta$ consists of transitions of the form
$(q,A,G) \hookrightarrow \{(q_1,\gamma_1,C_1),\dots,(q_k,\gamma_k,C_k)\}$ 
where $q,q_1,\dots,q_k \in Q$, $A \in \Gamma ,\gamma_1,\dots,\gamma_k \in \Gamma^*$, 
$G$ is a conjunction of elements from $\test$ and each $C_i$ is a 
consistent command. Intuitively, at a configuration $(q,A\gamma,F)$
the machine \emph{non-deterministically} selects a transition of the form
$(q,A,G) \hookrightarrow \{(q_1,\gamma_1,C_1),\dots,(q_k,\gamma_k,C_k)\}$
such that the assignment $F$ satisfies the conjunction $G$ and then \emph{forks} into $k$ copies in the configurations $(q_1,\gamma_1\gamma,F[C_1]),\dots,(q_k,\gamma_k\gamma,F[C_k])$
where $F[C_i]$ is the function obtained by updating $F$ according to the command $C_i$.
With this intuition in mind, we say that a transition $(q,A,G) \hookrightarrow \{(q_1,\gamma_1,C_1),\dots,(q_k,\gamma_k,C_k)\}$ is enabled at 
a configuration $(p,B\gamma,F)$ iff $p = q, B = A$ and $F$ satisfies all the 
tests in $G$.

A \emph{run} from a configuration $(q,\eta,H)$ to a configuration $(q',\eta',H')$ is a tree satisfying the following properties: The root is labelled by $(q,\eta,H)$. If some internal node $n$ is labelled by $(p,A\gamma,F)$ then there exists a transition 
$(p,A,G) \hookrightarrow \{(p_1,\gamma_1,C_1),(p_2,\gamma_2,C_2),\dots,(p_k,\gamma_k,C_k)\}$
which is enabled at $(p,A\gamma,F)$ such that the children of $n$ are labelled by $(p_1,\gamma_1\gamma,F[C_1])$,$\dots$, $(p_k,\gamma_k\gamma,F[C_k])$, where 
$F[C_i](v) = b$ if $C_i$ contains a command of the form $v \mapsto b$ and 
$F[C_i](v) = F(v)$ otherwise. Finally all the leaves
are labelled by $(q',\eta',H')$. If a run exists between $(q,\eta,H)$
and $(q',\eta',H')$ then we denote it by $(q,\eta,H) \xrightarrow[\Ac]{*} (q',\eta',H')$. 
An \emph{accepting run} from a configuration $(q,\eta,H)$ is a run
from $(q,\eta,H)$ to $(\fin,\bot,\mathbf{0})$ where $\mathbf{0}$ is the zero function.
An accepting run of an AEPS is simply an accepting run from the initial configuration $(\init,\bot,\mathbf{0})$. The emptiness problem is then to decide whether a given AEPS has an accepting run.

By a simple adaptation of the $\EXPTIME$-hardness proof for emptiness of 
alternating pushdown systems which have no Boolean variables (Theorem 5.4 of ~\cite{Alternation}, Prop. 31 of~\cite{PushdownGames})
we prove that
\begin{restatable}{lemma}{LemABPTwoEXP}\label{lemma:ABP-2EXP}
    The emptiness problem for AEPS is $\TWOEXPTIME$-hard.
\end{restatable}

An AEPS $\Ac$ is called a \emph{non-deterministic extended pushdown system} (NEPS) if every transition of $\Ac$ is of the form $(p,A,F) \hookrightarrow \{(q,\gamma,C)\}$. 
By Theorem 2 of~\cite{BooleanPrograms} we have that
\begin{restatable}{lemma}{LemmaNBPEXP}\label{lemma:NBP-EXP}
    The emptiness problem for NEPS is $\EXPTIME$-hard.
\end{restatable}

\begin{remark}
    The hardness result for AEPS could also be inferred from Theorem 10 of~\cite{BooleanPrograms}. Because we use a different notation, for the sake of completeness, we provide the proofs of both of these lemmas in the appendix.
\end{remark}

 \subsection{Reduction from Alternating Extended Pushdown Systems}

We now give a reduction from the emptiness problem for AEPS to {\scshape Special-Sync}.
Let $\Ac = (Q,V,\Gamma,\Delta,\init,\fin)$ be an AEPS. 
Without loss of generality, we can assume that 
if $(q,A,G) \hookrightarrow \{(q_1,\gamma_1,C_1),
(q_2,\gamma_2,C_2),\dots,(q_k,\gamma_k,C_k)\} \in \Delta$, then $\gamma_i \neq \gamma_j$
for $i \neq j$. (This can be accomplished, by prefixing new characters to each $\gamma_i$, moving to some intermediate states and then popping the new characters and moving to the respective $q_i$'s).
Having made this assumption, the reduction is described below.

From the given AEPS $\Ac$, we now construct a pushdown automaton $\Pc$ as follows.
The stack alphabet of $\Pc$ will be $\Gamma$.
For each transition $t \in \Delta$, $\Pc$ will have an input letter $\inputt(t)$.
$\Pc$ will also have another input letter $\End$.
The state space of $\Pc$ will be the set $Q \cup (V \times \{0,1\}) \cup \{q_{\acc},q_{\rej}\}$, where $q_{\acc}$ and $q_{\rej}$ are two states, which on reading any input letter, will leave the stack untouched and simply stay at $q_{\acc}$ and $q_{\rej}$ respectively. 
\begin{figure}
    \centering
    \makebox[\textwidth][c]{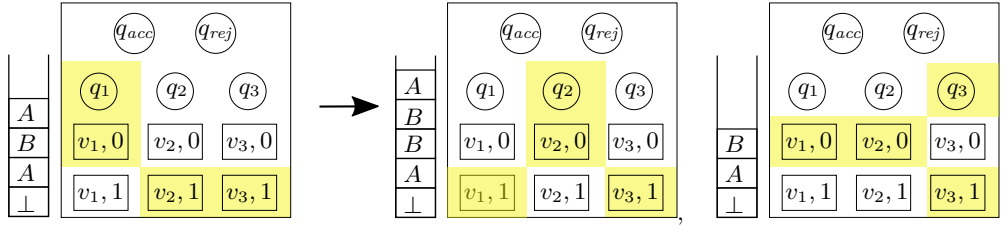}%
    \caption{Let $t$ be the transition $(q_1,A,[v_1?=0, v_3?=1]) \hookrightarrow \{(q_2, AB, [v_1\leftarrow 1, v_2\leftarrow 0]),(q_3,\epsilon, [v_2\leftarrow 0])\}$ in $\Ac$.  In $\Ac$, using $t$, the configuration $C_1 := (q_1,ABA\bot,[v_1=0,v_2=1,v_3=1])$ can
    fork into $C_2 := (q_2,ABBA\bot,[v_1=1,v_2=0,v_3=1])$ and 
    $(C_3 := q_3,BA\bot,[v_1=0,v_2=0,v_3=1])$. 
    The diagram here illustrates the
    simulation of this forking on the corresponding pseudo-configurations of $C_1,C_2,C_3$ that the automaton $\Pc$ will achieve when reading the letter $\inputt(t)$. 
    \label{fig:AEPStoPDA}}
\end{figure}

We now give an intuition behind the transitions of $\Pc$. Given an assignment $F : V \to \{0,1\}$ of the Boolean variables $V$, and a state $q$ of $\Ac$, we use 
the notation $\set{q,F}$ to denote the subset $\{q\} \cup \{(v,F(v)) : v \in V\}$ of states of $\Pc$. Intuitively, a configuration $(q,\gamma,F)$ of 
$\Ac$ is simulated by its corresponding \emph{pseudo-configuration} $(\set{q,F},\gamma)$ in $\Pc$.
This intuition is captured by Figure~\ref{fig:AEPStoPDA}, which 
gives an example of a step in $\Pc$.

Now we give a formal description of the transitions of $\Pc$. 
Let $t = (q,A,G) \hookrightarrow \{(q_1,\gamma_1,C_1),\dots,(q_k,\gamma_k,C_k)\}$ be
a transition of $\Ac$.  Let $p \in Q$. Upon reading $\inputt(t)$, if $p \neq q$ then $p$ immediately 
moves to the $q_\rej$ state. Further, even state $q$ moves to the $q_\rej$ state
if the top of the stack is not $A$. However, if the top of the 
stack is $A$, then $q$ pops $A$ and \emph{non-deterministically} pushes 
any one of $\gamma_1,\dots,\gamma_k$ onto the stack and if it pushed $\gamma_i$, then $q$
moves to the state $q_i$.

Let $(v,b) \in V \times \{0,1\}$. Upon reading $\inputt(t)$, if the test $v ?= (1-b)$ appears in the guard $G$, then $(v,b)$ immediately moves to the $q_{\rej}$ state. (Notice
that this is a purely syntactical condition on $\Ac$). 
Further, if the top of the stack is not $A$, then once again $(v,b)$ moves to $q_{\rej}$.
If these two cases do not hold, then $(v,b)$ pops $A$ and \emph{non-deterministically} picks an $i\in \{1,\dots k\}$ and pushes $\gamma_i$ onto the stack. Having pushed $\gamma_i$, if $C_i$ does not update the variable $v$,  it stays in state $(v,b)$; otherwise if $C_i$ has a command $v \mapsto b'$, it moves to $(v,b')$.

Finally, upon reading $\End$, the states in $\set{\fin,\zero}$ move to the 
$q_{\acc}$ state and all the other states in $Q \cup (V \times \{0,1\})$ move to the $q_{\rej}$ state. 

We now claim that $\Ac$ has an accepting run iff there
is a super-synchronizer in $\Pc$ between $(\set{\init,\zero},\bot)$ and $q_{\acc}$.
Intuitively, any accepting run of $\Ac$ can be simulated by the corresponding pseudo-configurations in a manner similar to Figure~\ref{fig:AEPStoPDA} and once we arrive at
the pseudo-configuration $(\set{\fin,\zero},\bot)$,
we can input the letter $\End$ and synchronise to the state $q_{\acc}$.
For the reverse direction, we can show that any super-synchroniser between
$(\set{\init,\zero},\bot)$ and $q_{\acc}$ must be a simulation of 
an accepting run in $\Ac$.

Notice that $\Pc$ is deterministic if $\Ac$ is non-deterministic.  
Hence, by Lemmas~\ref{lemma:ABP-2EXP} and \ref{lemma:NBP-EXP},
\begin{restatable}{theorem}{ThmSyncHardness}\label{thm:sync-hardness}
    {\scshape Special-Sync}, {\scshape Subset-Ada-Sync} and {\scshape Ada-Sync} are all $\TWOEXPTIME$-hard. {\scshape Det-Special-Sync}, {\scshape Det-Subset-Ada-Sync} and {\scshape Det-Ada-Sync} are all $\EXPTIME$-hard.
\end{restatable}

%% file: lb.pdf_tex
\begingroup%
  \makeatletter%
  \providecommand\color[2][]{%
    \errmessage{(Inkscape) Color is used for the text in Inkscape, but the package 'color.sty' is not loaded}%
    \renewcommand\color[2][]{}%
  }%
  \providecommand\transparent[1]{%
    \errmessage{(Inkscape) Transparency is used (non-zero) for the text in Inkscape, but the package 'transparent.sty' is not loaded}%
    \renewcommand\transparent[1]{}%
  }%
  \providecommand\rotatebox[2]{#2}%
  \newcommand*\fsize{\dimexpr\f@size pt\relax}%
  \newcommand*\lineheight[1]{\fontsize{\fsize}{#1\fsize}\selectfont}%
  \ifx\svgwidth\undefined%
    \setlength{\unitlength}{369.70047141bp}%
    \ifx\svgscale\undefined%
      \relax%
    \else%
      \setlength{\unitlength}{\unitlength * \real{\svgscale}}%
    \fi%
  \else%
    \setlength{\unitlength}{\svgwidth}%
  \fi%
  \global\let\svgwidth\undefined%
  \global\let\svgscale\undefined%
  \makeatother%
  \begin{picture}(1,0.22346347)%
    \lineheight{1}%
    \setlength\tabcolsep{0pt}%
    \put(0.04426207,0.2030414){\color[rgb]{0,0,0}\makebox(0,0)[lt]{\begin{minipage}{1.26634931\unitlength}\raggedright \end{minipage}}}%
    \put(0,0){\includegraphics[width=\unitlength,page=1]{lb.pdf}}%
    \put(0.67475739,0.00238527){\color[rgb]{0,0,0}\makebox(0,0)[lt]{\lineheight{1.25}\smash{\begin{tabular}[t]{l}$,$\end{tabular}}}}%
    \put(0,0){\includegraphics[width=\unitlength,page=2]{lb.pdf}}%
    \put(0.08214718,0.1306157){\color[rgb]{0,0,0}\makebox(0,0)[lt]{\lineheight{1.25}\smash{\begin{tabular}[t]{l}$q_1$\end{tabular}}}}%
    \put(0,0){\includegraphics[width=\unitlength,page=3]{lb.pdf}}%
    \put(0.15922006,0.13087186){\color[rgb]{0,0,0}\makebox(0,0)[lt]{\lineheight{1.25}\smash{\begin{tabular}[t]{l}$q_2$\end{tabular}}}}%
    \put(0,0){\includegraphics[width=\unitlength,page=4]{lb.pdf}}%
    \put(0.23252248,0.13087189){\color[rgb]{0,0,0}\makebox(0,0)[lt]{\lineheight{1.25}\smash{\begin{tabular}[t]{l}$q_3$\end{tabular}}}}%
    \put(0,0){\includegraphics[width=\unitlength,page=5]{lb.pdf}}%
    \put(0.06980451,0.07543288){\color[rgb]{0,0,0}\makebox(0,0)[lt]{\lineheight{1.25}\smash{\begin{tabular}[t]{l}$v_1,0$\end{tabular}}}}%
    \put(0.14693726,0.07546394){\color[rgb]{0,0,0}\makebox(0,0)[lt]{\lineheight{1.25}\smash{\begin{tabular}[t]{l}$v_2,0$\end{tabular}}}}%
    \put(0.21983661,0.07516382){\color[rgb]{0,0,0}\makebox(0,0)[lt]{\lineheight{1.25}\smash{\begin{tabular}[t]{l}$v_3,0$\end{tabular}}}}%
    \put(0,0){\includegraphics[width=\unitlength,page=6]{lb.pdf}}%
    \put(0.1082587,0.19038295){\color[rgb]{0,0,0}\makebox(0,0)[lt]{\lineheight{1.25}\smash{\begin{tabular}[t]{l}$q_{\acc}$\end{tabular}}}}%
    \put(0.19172513,0.19038675){\color[rgb]{0,0,0}\makebox(0,0)[lt]{\lineheight{1.25}\smash{\begin{tabular}[t]{l}$q_{\rej}$\end{tabular}}}}%
    \put(0,0){\includegraphics[width=\unitlength,page=7]{lb.pdf}}%
    \put(0.06980638,0.02348598){\color[rgb]{0,0,0}\makebox(0,0)[lt]{\lineheight{1.25}\smash{\begin{tabular}[t]{l}$v_1,1$\end{tabular}}}}%
    \put(0.14693912,0.02351704){\color[rgb]{0,0,0}\makebox(0,0)[lt]{\lineheight{1.25}\smash{\begin{tabular}[t]{l}$v_2,1$\end{tabular}}}}%
    \put(0.21983848,0.02321695){\color[rgb]{0,0,0}\makebox(0,0)[lt]{\lineheight{1.25}\smash{\begin{tabular}[t]{l}$v_3,1$\end{tabular}}}}%
    \put(0,0){\includegraphics[width=\unitlength,page=8]{lb.pdf}}%
    \put(0.47187534,0.13061567){\color[rgb]{0,0,0}\makebox(0,0)[lt]{\lineheight{1.25}\smash{\begin{tabular}[t]{l}$q_1$\end{tabular}}}}%
    \put(0,0){\includegraphics[width=\unitlength,page=9]{lb.pdf}}%
    \put(0.54894824,0.13087183){\color[rgb]{0,0,0}\makebox(0,0)[lt]{\lineheight{1.25}\smash{\begin{tabular}[t]{l}$q_2$\end{tabular}}}}%
    \put(0,0){\includegraphics[width=\unitlength,page=10]{lb.pdf}}%
    \put(0.6222506,0.13087189){\color[rgb]{0,0,0}\makebox(0,0)[lt]{\lineheight{1.25}\smash{\begin{tabular}[t]{l}$q_3$\end{tabular}}}}%
    \put(0,0){\includegraphics[width=\unitlength,page=11]{lb.pdf}}%
    \put(0.45953268,0.07543285){\color[rgb]{0,0,0}\makebox(0,0)[lt]{\lineheight{1.25}\smash{\begin{tabular}[t]{l}$v_1,0$\end{tabular}}}}%
    \put(0.53666542,0.07546391){\color[rgb]{0,0,0}\makebox(0,0)[lt]{\lineheight{1.25}\smash{\begin{tabular}[t]{l}$v_2,0$\end{tabular}}}}%
    \put(0.60956473,0.07516382){\color[rgb]{0,0,0}\makebox(0,0)[lt]{\lineheight{1.25}\smash{\begin{tabular}[t]{l}$v_3,0$\end{tabular}}}}%
    \put(0,0){\includegraphics[width=\unitlength,page=12]{lb.pdf}}%
    \put(0.49798685,0.19038295){\color[rgb]{0,0,0}\makebox(0,0)[lt]{\lineheight{1.25}\smash{\begin{tabular}[t]{l}$q_{\acc}$\end{tabular}}}}%
    \put(0.58145325,0.19038675){\color[rgb]{0,0,0}\makebox(0,0)[lt]{\lineheight{1.25}\smash{\begin{tabular}[t]{l}$q_{\rej}$\end{tabular}}}}%
    \put(0,0){\includegraphics[width=\unitlength,page=13]{lb.pdf}}%
    \put(0.45953455,0.02348598){\color[rgb]{0,0,0}\makebox(0,0)[lt]{\lineheight{1.25}\smash{\begin{tabular}[t]{l}$v_1,1$\end{tabular}}}}%
    \put(0.53666729,0.02351698){\color[rgb]{0,0,0}\makebox(0,0)[lt]{\lineheight{1.25}\smash{\begin{tabular}[t]{l}$v_2,1$\end{tabular}}}}%
    \put(0.6095666,0.02321689){\color[rgb]{0,0,0}\makebox(0,0)[lt]{\lineheight{1.25}\smash{\begin{tabular}[t]{l}$v_3,1$\end{tabular}}}}%
    \put(0,0){\includegraphics[width=\unitlength,page=14]{lb.pdf}}%
    \put(0.79668626,0.1306157){\color[rgb]{0,0,0}\makebox(0,0)[lt]{\lineheight{1.25}\smash{\begin{tabular}[t]{l}$q_1$\end{tabular}}}}%
    \put(0,0){\includegraphics[width=\unitlength,page=15]{lb.pdf}}%
    \put(0.87375834,0.13087186){\color[rgb]{0,0,0}\makebox(0,0)[lt]{\lineheight{1.25}\smash{\begin{tabular}[t]{l}$q_2$\end{tabular}}}}%
    \put(0,0){\includegraphics[width=\unitlength,page=16]{lb.pdf}}%
    \put(0.78434337,0.07543285){\color[rgb]{0,0,0}\makebox(0,0)[lt]{\lineheight{1.25}\smash{\begin{tabular}[t]{l}$v_1,0$\end{tabular}}}}%
    \put(0.86147587,0.07546391){\color[rgb]{0,0,0}\makebox(0,0)[lt]{\lineheight{1.25}\smash{\begin{tabular}[t]{l}$v_2,0$\end{tabular}}}}%
    \put(0.93437519,0.07516379){\color[rgb]{0,0,0}\makebox(0,0)[lt]{\lineheight{1.25}\smash{\begin{tabular}[t]{l}$v_3,0$\end{tabular}}}}%
    \put(0,0){\includegraphics[width=\unitlength,page=17]{lb.pdf}}%
    \put(0.82279678,0.19038295){\color[rgb]{0,0,0}\makebox(0,0)[lt]{\lineheight{1.25}\smash{\begin{tabular}[t]{l}$q_{\acc}$\end{tabular}}}}%
    \put(0.9062637,0.19038675){\color[rgb]{0,0,0}\makebox(0,0)[lt]{\lineheight{1.25}\smash{\begin{tabular}[t]{l}$q_{\rej}$\end{tabular}}}}%
    \put(0,0){\includegraphics[width=\unitlength,page=18]{lb.pdf}}%
    \put(0.78434506,0.02348598){\color[rgb]{0,0,0}\makebox(0,0)[lt]{\lineheight{1.25}\smash{\begin{tabular}[t]{l}$v_1,1$\end{tabular}}}}%
    \put(0.86147768,0.02351698){\color[rgb]{0,0,0}\makebox(0,0)[lt]{\lineheight{1.25}\smash{\begin{tabular}[t]{l}$v_2,1$\end{tabular}}}}%
    \put(0.93437706,0.02321689){\color[rgb]{0,0,0}\makebox(0,0)[lt]{\lineheight{1.25}\smash{\begin{tabular}[t]{l}$v_3,1$\end{tabular}}}}%
    \put(0,0){\includegraphics[width=\unitlength,page=19]{lb.pdf}}%
    \put(0.00844748,0.01056119){\color[rgb]{0,0,0}\makebox(0,0)[lt]{\lineheight{1.25}\smash{\begin{tabular}[t]{l}$\bot$\end{tabular}}}}%
    \put(0.00845065,0.04192185){\color[rgb]{0,0,0}\makebox(0,0)[lt]{\lineheight{1.25}\smash{\begin{tabular}[t]{l}$A$\end{tabular}}}}%
    \put(0.00845065,0.07197627){\color[rgb]{0,0,0}\makebox(0,0)[lt]{\lineheight{1.25}\smash{\begin{tabular}[t]{l}$B$\end{tabular}}}}%
    \put(0.00845065,0.10333692){\color[rgb]{0,0,0}\makebox(0,0)[lt]{\lineheight{1.25}\smash{\begin{tabular}[t]{l}$A$\end{tabular}}}}%
    \put(0,0){\includegraphics[width=\unitlength,page=20]{lb.pdf}}%
    \put(0.39879346,0.00959514){\color[rgb]{0,0,0}\makebox(0,0)[lt]{\lineheight{1.25}\smash{\begin{tabular}[t]{l}$\bot$\end{tabular}}}}%
    \put(0.39879661,0.04095581){\color[rgb]{0,0,0}\makebox(0,0)[lt]{\lineheight{1.25}\smash{\begin{tabular}[t]{l}$A$\end{tabular}}}}%
    \put(0.39879661,0.07101022){\color[rgb]{0,0,0}\makebox(0,0)[lt]{\lineheight{1.25}\smash{\begin{tabular}[t]{l}$B$\end{tabular}}}}%
    \put(0,0){\includegraphics[width=\unitlength,page=21]{lb.pdf}}%
    \put(0.72309861,0.00959516){\color[rgb]{0,0,0}\makebox(0,0)[lt]{\lineheight{1.25}\smash{\begin{tabular}[t]{l}$\bot$\end{tabular}}}}%
    \put(0.72310182,0.04095581){\color[rgb]{0,0,0}\makebox(0,0)[lt]{\lineheight{1.25}\smash{\begin{tabular}[t]{l}$A$\end{tabular}}}}%
    \put(0.72310182,0.07101023){\color[rgb]{0,0,0}\makebox(0,0)[lt]{\lineheight{1.25}\smash{\begin{tabular}[t]{l}$B$\end{tabular}}}}%
    \put(0,0){\includegraphics[width=\unitlength,page=22]{lb.pdf}}%
    \put(0.39879857,0.09933699){\color[rgb]{0,0,0}\makebox(0,0)[lt]{\lineheight{1.25}\smash{\begin{tabular}[t]{l}$B$\end{tabular}}}}%
    \put(0.39879857,0.13069762){\color[rgb]{0,0,0}\makebox(0,0)[lt]{\lineheight{1.25}\smash{\begin{tabular}[t]{l}$A$\end{tabular}}}}%
    \put(0,0){\includegraphics[width=\unitlength,page=23]{lb.pdf}}%
    \put(0.9470593,0.13087315){\color[rgb]{0,0,0}\makebox(0,0)[lt]{\lineheight{1.25}\smash{\begin{tabular}[t]{l}$q_3$\end{tabular}}}}%
  \end{picture}%
\endgroup%

%% file: upperbound_2EXP.tex
In this section, we will give algorithms that solve {\scshape{Special-Sync}} and {\scshape{Det-Special-Sync}}. We first give a reduction from {\scshape{Special-Sync}} to the problem of checking emptiness in an alternating pushdown system, which we define below. Then, we show that for {\scshape Det-Special-Sync}, the same
reduction produces alternating pushdown systems with a ``modular'' structure,
which we exploit to reduce the running time.
\subsection{Adaptive Synchronisation for Non-deterministic PDA}\label{subsec:SpecialSync2EXP}
An alternating pushdown system (APS) is an alternating extended pushdown system which has no Boolean variables. Since there are no variables, we can suppress any notation corresponding to the variables, e.g., configurations can be just denoted by $(q,\gamma)$.
It is known that the emptiness problem for APS is in $\EXPTIME$ (Theorem 4.1 of~\cite{ReachabilityAnalysis}).
We now give an exponential time reduction from {\scshape Special-Sync} to the emptiness problem for APS. 

Let $\Pc = (Q,\Sigma,\Gamma,\delta)$ be a PDA with $I \subseteq Q$, $s \in Q$. Construct the following APS $\Ac_{\Pc} = (2^Q, \Gamma, \Delta, I, \{s\})$ where $\Delta$ is defined as follows:
Given $S \subseteq Q, a \in \Sigma$ and $A \in \Gamma$, let 
$E_1,\dots,E_k$ be the equivalence classes of the relation $\sim^a_{S,A}$ as defined in subsection~\ref{subsec:define-ada-sync}. Then, we have the following transition in $\Ac_{\Pc}$:
\begin{equation}~\label{eq:transitionsA_P}
  (S,A) \hookrightarrow \{(\nex(E_1),\word(E_1)),(\nex(E_2),\word(E_2)),\dots,(\nex(E_k),\word(E_k)\}  
\end{equation}

The following fact is immediate from the definition of a super-synchroniser and from the construction of $\Ac_{\Pc}$. 
\begin{proposition}\label{prop:correctnessA_P}
    Let $S \subseteq 2^Q, \gamma \in \Gamma^*$.
    Then a labelled tree $T$ is a super-synchroniser between $(S,\gamma)$ and $s$ in $\Pc$ if and only if $T$ is an accepting run from $(S,\gamma)$ in $\Ac_{\Pc}$.
\end{proposition}

By Theorem 4.1 of~\cite{ReachabilityAnalysis}, emptiness
for APS can be solved in exponential time and so 
\begin{theorem}
    {\scshape Special-Sync} is in $\TWOEXPTIME$
\end{theorem}

%% file: upperbound_EXP_det.tex
\subsection{Adaptive Synchronisation for Deterministic PDA}
Let $\Pc = (Q,\Sigma,\Gamma,\delta)$ be a deterministic PDA with $I \subseteq Q, s \in Q$. We have the following proposition, whose proof follows from the fact that $\Pc$ is deterministic.
\begin{restatable}{proposition}{PropEquiClassesSmallSize}~\label{prop:equi-classes-small-size}
    Suppose $S \subseteq Q, a \in \Sigma, A \in \Gamma$ and 
    suppose $E_1,\dots,E_k$ are the equivalence classes of $\sim^a_{S,A}$.
    Then, $|S| \ge \sum_{i=1}^k |\nex(E_i)|$.
\end{restatable}

Now, given $\Pc$, consider the APM $\Ac_{\Pc} = (2^Q,\Gamma,\Delta,I,\{s\})$ that we have constructed in subsection~\ref{subsec:SpecialSync2EXP}.  By Proposition~\ref{prop:equi-classes-small-size}, we now have the following lemma.
\begin{restatable}{lemma}{LemSmallLeaves}
    For any $S \in 2^Q, \gamma \in \Gamma^*$, any accepting run of $\Ac_{\Pc}$
    from the configuration $(S,\gamma)$ has at most $|S|$ leaves.
\end{restatable}

The following corollary follows from the lemma above.
\begin{corollary}~\label{cor:limited-leaves}
    Any accepting run of $\Ac_{\Pc}$ has at most $|I|$ leaves.
\end{corollary}

\begin{example}
    Let $\Pc$ be the deterministic PDA from Figure~\ref{fig:automata}.
    Figure~\ref{fig:run} shows an example of an accepting run in the corresponding
    APS $\Ac_\Pc$ from $I := \{1,2,3,4\}$. Notice that 
    there are $|I| = 4$ leaves in this run.
\end{example}

Corollary~\ref{cor:limited-leaves} motivates the study of the following problem, which we call the 
\emph{sparse emptiness} problem for APMs ({\scshape Sparse-Empty}): 
\begin{quote}
    \emph{Given: } An APM $\Ac$ and a number $k$ in unary.\\
    \emph{Decide: } Whether there exists an accepting run for $\Ac$ with at most $k$ leaves
\end{quote}
We prove the following theorem about {\scshape{Sparse-Empty}} in the next section. 
\begin{theorem}\label{thm:Sparse-Empty-Exp}
    Given $\Ac$ and $k$, the {\scshape{Sparse-Empty}} problem can be solved in time $O(|\Ac|^{ck^2})$
    for a fixed constant $c$.
\end{theorem}

Now, because of Proposition~\ref{prop:equi-classes-small-size} and because of the structure of the transitions of $\Ac_{\Pc}$ (as given by equation~(\ref{eq:transitionsA_P})), it is sufficient to restrict the construction of $\Ac_{\Pc}$ to only those states which have cardinality at most $|I|$ and hence, it can be assumed that $|\Ac_{\Pc}| \le |\Pc|^{4|I|}$. This fact, along with Proposition~\ref{prop:correctnessA_P}, corollary~\ref{cor:limited-leaves} and Theorem~\ref{thm:Sparse-Empty-Exp} implies the following theorem.
\begin{theorem}\label{thm:Det-Special-Sync-EXP}
    Given an instance $(\Pc,I,s)$ of {\scshape Det-Special-Sync}, we can check
    if $(I,\bot) \xRightarrow[\Pc]{\text{sup}} s$ in time $O(n^{4ck^3})$ 
    where $n = |\Pc|$ and $k = |I|$ and $c$ is some fixed constant.
\end{theorem}
\begin{remark}
    Note that the algorithm to solve {\scshape Det-Special-Sync} on an instance $(\Pc,I,s)$, although in $\EXPTIME$, is polynomial if $|I|$ is fixed 
    and quasi-polynomial if $|I| = O(\log |\Pc|)$.
\end{remark}

\subsection{`Sparse Emptiness' Checking of Alternating Systems}~\label{subsec:Sparse-Empty}
This subsection is dedicated to proving Theorem~\ref{thm:Sparse-Empty-Exp}. We fix an alternating pushdown system $\Ac = (Q,\Gamma,\Delta,\init,\fin)$ and a number $k$ for the rest of this subsection. A $k$-accepting run of $\Ac$ is defined to be an accepting run of $\Ac$ with at most $k$ leaves.
We now split the desired algorithm for {\scshape Sparse-Empty} into three parts. Finally, we give its runtime analysis.

\textbf{Compressing $k$-accepting runs of $\Ac$:}
We define a non-deterministic pushdown system (NPS) to be a non-deterministic extended pushdown system which has no Boolean variables. From $\Ac$, we can derive a NPS
obtained by deleting all transitions of the form $(q,A) \hookrightarrow \{(q_1,  \gamma_1),\dots,(q_k,\gamma_k)\}$ with $k > 1$. 
We will denote this NPS by $\Nc$.
Emptiness of NPS is known to be solvable in polynomial time (Theorem 2.1 of~\cite{ReachabilityAnalysis}). To exploit this fact for our problem, we propose the following notion of a \emph{compressed} accepting run of $\Ac$. 
Intuitively, a compressed accepting run is obtained from an accepting run of $\Ac$ 
by ``compressing'' a series of transitions belonging to the non-deterministic part $\Nc$,
into a single transition.
An intuition of a compressed accepting run is captured by Figure~\ref{fig:comprun}, which is obtained by compressing the run depicted in Figure~\ref{fig:run}.
\begin{figure}
\begin{minipage}{0.45 \textwidth}
    \centering
\makebox[\textwidth][c]{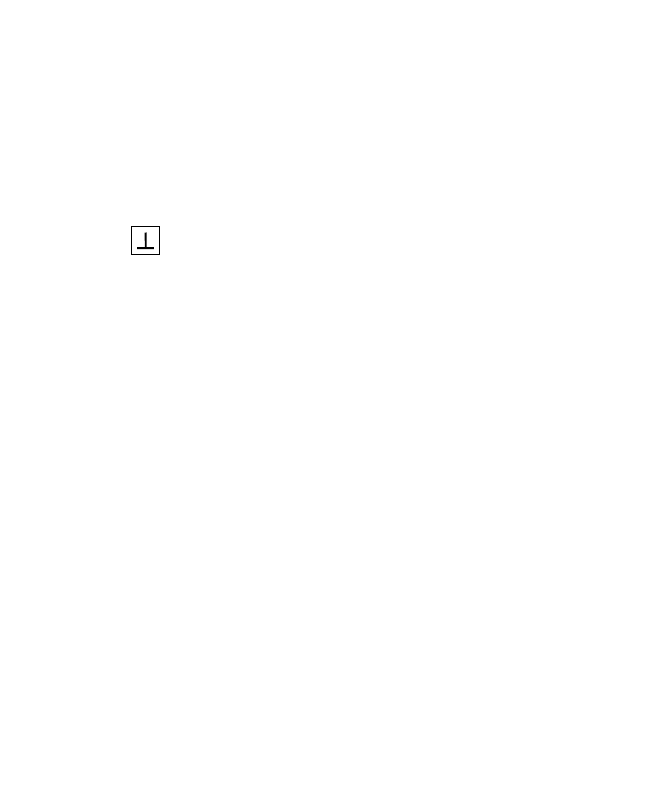}%
\caption{An accepting run of $\Ac_\Pc$ for the deterministic PDA $\Pc$ given in Figure~\ref{fig:automata}.
\label{fig:run}}
\end{minipage}
\hfill \vline \hfill
\begin{minipage}{0.40 \textwidth}
   \centering
\makebox[\textwidth][c]{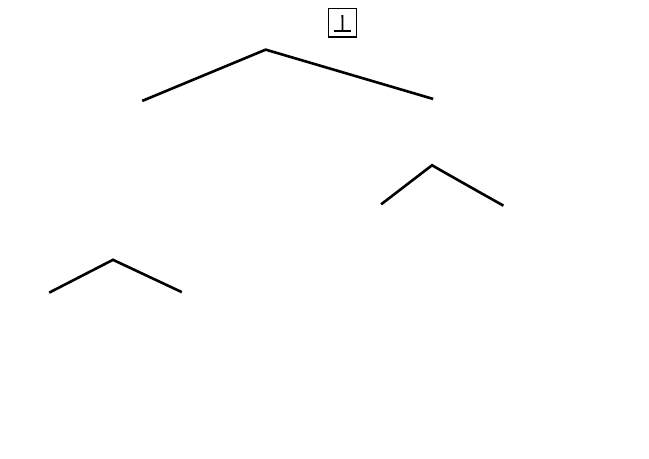}%
\caption{A compressed accepting run of $\Ac_\Pc$ for the deterministic PDA $\Pc$ given
in Figure~\ref{fig:automata}, obtained by compressing the run from Figure~\ref{fig:run}
\label{fig:comprun}}
\end{minipage}
\end{figure}

Given a tree, we say that a vertex $v$ in the tree is \emph{simple} if it has exactly one
child and otherwise we say that it is \emph{complex} (Note that all leaves are complex). A \emph{compressed accepting run} of $\Ac$ from the configuration $(p,\eta)$ is a labelled tree such that: The root is labelled by $(p,\eta)$. If
$v$ is a simple vertex labelled by $(q,\gamma)$ and $u$ is its only child labelled by $(q',\gamma')$ then $u$ is a complex vertex and $(q,\gamma)  \xrightarrow[\Nc]{*} (q',\gamma')$.
If $v$ is a complex vertex labelled by $(q,A\gamma)$ and $v_1,\dots,v_k$ are its children with $k > 1$, then there is a transition $(q,A) \hookrightarrow \{(q_1,A_1),\dots,(q_k,A_k)\}$ in $\Ac$ such that the label of $v_i$ is $(q_i,A_i\gamma)$. Finally, all the leaves are labelled by $(\fin,\bot)$.
A compressed accepting run of $\Ac$ is a compressed accepting run from $(\init,\bot)$ and a $k$-compressed accepting run is a compressed accepting run with at most $k$ leaves. 
We now have the following lemma.
\begin{restatable}{lemma}{LemCompress}~\label{lem:compress}
    There is a $k$-accepting run of $\Ac$ from a configuration $(p,\eta)$ iff 
    there is a $k$-compressed accepting run of $\Ac$ from $(p,\eta)$.
\end{restatable}

\textbf{Searching for $k$-compressed accepting runs:}
To fully use the result of Lemma~\ref{lem:compress}, we need some 
results about non-deterministic pushdown systems, which we state here.
Recall that $\Nc$ is an NPS over the states $Q$ and stack alphabet $\Gamma$ obtained from the APS $\Ac$. We say that $M = (Q^M,\Gamma,\delta^M,F^M)$ is an \emph{$\Nc$-automaton} if $M$
  is a non-det. finite-state automaton over the alphabet $\Gamma$ with accepting states $F^M$ such that for each
  state $q \in Q$, there is a unique state $q^M \in Q^M$.
    The set of configurations of $\Ac$ that are stored by $M$ (denoted by $\Cc(M)$) is defined to be the set $\{(q,\gamma) : \gamma \text{ is accepted in } M \text{ from the state } q^M \}$. In the above definition, note that $Q^M$ can potentially have more states other than the set $\{q^M\mid q \in Q\}$.
\begin{theorem}~\label{th:prestar}(Section 2.3 and Theorem 2.1 of~\cite{ReachabilityAnalysis})
    Given an $\Nc$-automaton $M$,
    in time polynomial in $\Nc$ and $M$, we can construct an $\Nc$-automaton $M'$ which has
    the \emph{same states as $M$} such that $M'$ stores the set of predecessors of $M$, i.e.,
    $\Cc(M') = \{(q',\gamma') : \exists (q,\gamma) \in \Cc(M) \text{ such that } (q',\gamma') \xrightarrow[\Nc]{*} (q,\gamma)\}$.
\end{theorem}

We say that an unlabelled tree is \emph{structured}, if the child of every simple vertex
is a complex vertex. An $\ell$-structured tree is simply a structured tree which has at most $\ell$ leaves. Notice that the height of an $\ell$-structured tree is $O(\ell)$ and since it has at most $\ell$ leaves, it follows that a $\ell$-structured tree can be described using a polynomial
number of bits in $\ell$. Hence, the number of $\ell$-structured trees is $O(2^{\ell^c})$ 
for some fixed $c$.

Now let us come back to the problem of searching for $k$-accepting runs of $\Ac$.
By Lemma~\ref{lem:compress} it suffices to search for a $k$-compressed accepting
run of $\Ac$. Notice that if we take a $k$-compressed accepting run and 
remove its labels, we get a $k$-structured tree. Now, suppose we have
an algorithm $\mathtt{Check}$ that takes a $k$-structured tree $T$ and checks if $T$ can be labelled to make it a $k$-compressed accepting run of $\Ac$.
Then, by calling $\mathtt{Check}$ on every $k$-structured tree,
we have an algorithm to check for the existence of a $k$-compressed
accepting run of $\Ac$. Hence, it suffices to describe 
this procedure $\mathtt{Check}$ which is what we will do now.

\textbf{The algorithm $\mathtt{Check}$:}
Let $T$ be a $k$-structured tree. For each vertex $v$ in the tree $T$, $\mathtt{Check}$ will assign a $\Nc$-automaton $M_v$ such that $M_v$ will have
the following property: 
\begin{quote}\label{invariant}
    Invariant (*) : A configuration $(q,\gamma) \in \Cc(M_v)$ iff all the vertices of the subtree rooted at $v$ can be labelled such that
    the resulting labelled subtree is a compressed accepting run of $\Ac$ from $(q,\gamma)$.
\end{quote}

The construction of each $M_v$ is as follows:
Let $Q$ be the states and $\Delta$ be the transitions of the alternating pushdown system $\Ac$. 
\begin{itemize}
    \item Suppose vertex $v$ is a leaf. We let $M_v$ be an automaton such that $\Cc(M_v) = \{(\fin,\bot)\}$. Notice that such a $M_v$ can be easily constructed in polynomial time.
    \item Suppose vertex $v$ is simple and $u$ is its child. We take $M_{u}$ and use Theorem~\ref{th:prestar} to construct the $\Nc$-automaton $M_v$. 
    Note that $M_v$ has the same set of states as $M_{u}$.
    \item Suppose $v$ is complex and suppose $v_1,\dots,v_\ell$ are its children. For each $1 \le i \le \ell$ and for every configuration $(q,\gamma)$ of $\Ac$, let $\delta_i(q^{M_{v_i}},\gamma)$ denote the set of states that the automaton $M_{v_i}$ will be in after reading $\gamma$ from the state $q^{M_{v_i}}$. To construct $M_v$ first do a product construction which we denote by $M_{v_1} \times M_{v_2} \times \dots \times M_{v_\ell}$. Then, for each $q \in Q$, add a state $q^{M_v}$.
    Then for each transition $(p,A) \hookrightarrow \{(p_1,\gamma_1),\dots,(p_\ell,\gamma_\ell)\}$ in $\Delta$, add a transition in $M_v$,
    which upon reading $A$, takes $p^{M_v}$ to any of the states in $\delta_1({p_1}^{M_{v_1}},\gamma_1) \times \delta_2({p_2}^{M_{v_2}},\gamma_2) \times \dots \times \delta_l({p_\ell}^{M_{v_\ell}},\gamma_\ell)$. Intuitively, we accept a word $A\gamma$ from the state $p^{M_v}$ if for each $i$, the
    word $\gamma_i\gamma$ can be accepted from the state ${p_i}^{M_{v_i}}$. 
\end{itemize}
 
\begin{restatable}{proposition}{PropInvariant}
    For each vertex $v$ of the tree $T$, $M_v$ satisfies invariant (*)
\end{restatable}
Finally, we accept iff $(\init,\bot) \in \Cc(M_r)$ where $r$ is the root
of the tree. The correctness of $\mathtt{Check}$ follows from the proposition above.

\paragraph*{Running time analysis}
Let us analyse the running time of $\mathtt{Check}$.
Let $T$ be a $k$-structured tree and therefore $T$ has $O\left(k^2\right)$ vertices.
$\mathtt{Check}$ assigns to each vertex $v$ of $T$ an automaton $M_v$.
We claim that the running time of $\mathtt{Check}$ is $O\left(k^2 \cdot |\Ac|^{ck^2}\right)$ (for some constant $c$) because
of the following: 
\begin{itemize}
    \item[1)] By induction on the structure of the tree $T$, it can be proved that, there exists a constant $d$, such that if $h_v$ is the height of 
a vertex $v$ and $l_v$ is the number of leaves in the subtree of $v$, then the 
number of states of $M_v$ is $O\left(|\Ac|^{dh_vl_v}\right)$ (Recall that $h_vl_v$ is at most $O\left(k^2\right)$). 
    \item[2)] If an $\Nc$-automaton has $n$ states,
then the number of transitions it can have is $O\left(|\Ac| \cdot n^2\right)$. 
    \item[3)] For a vertex $v$
with children $v_1,\dots,v_\ell$, $M_v$ can be constructed in polynomial time
in the size of $|M_{v_1}| \times |M_{v_2}| \times \dots |M_{v_\ell}|$ and $|\Ac|$.
\end{itemize}

Now the final algorithm for {\scshape Sparse-Empty} simply iterates over all $k$-structured trees
and calls $\mathtt{Check}$ on all of them. Since the
number of $k$-structured trees is at most $f(k)$ where $f$ is an exponential function, it follows that the total running time is $O\left(f(k) \cdot k^2 \cdot |\Ac|^{ck^2}\right) = O\left(|\Ac|^{ek^2}\right)$ for some constant $e$.

%% file: comp.pdf_tex
\begingroup%
  \makeatletter%
  \providecommand\color[2][]{%
    \errmessage{(Inkscape) Color is used for the text in Inkscape, but the package 'color.sty' is not loaded}%
    \renewcommand\color[2][]{}%
  }%
  \providecommand\transparent[1]{%
    \errmessage{(Inkscape) Transparency is used (non-zero) for the text in Inkscape, but the package 'transparent.sty' is not loaded}%
    \renewcommand\transparent[1]{}%
  }%
  \providecommand\rotatebox[2]{#2}%
  \newcommand*\fsize{\dimexpr\f@size pt\relax}%
  \newcommand*\lineheight[1]{\fontsize{\fsize}{#1\fsize}\selectfont}%
  \ifx\svgwidth\undefined%
    \setlength{\unitlength}{189.29117998bp}%
    \ifx\svgscale\undefined%
      \relax%
    \else%
      \setlength{\unitlength}{\unitlength * \real{\svgscale}}%
    \fi%
  \else%
    \setlength{\unitlength}{\svgwidth}%
  \fi%
  \global\let\svgwidth\undefined%
  \global\let\svgscale\undefined%
  \makeatother%
  \begin{picture}(1,1.22834319)%
    \lineheight{1}%
    \setlength\tabcolsep{0pt}%
    \put(0,0){\includegraphics[width=\unitlength,page=1]{comp.pdf}}%
    \put(0.07575325,0.84778378){\color[rgb]{0,0,0}\makebox(0,0)[lt]{\lineheight{1.25}\smash{\begin{tabular}[t]{l}$\{1,2\}$\end{tabular}}}}%
    \put(0.25223872,0.56391167){\color[rgb]{0,0,0}\makebox(0,0)[lt]{\lineheight{1.25}\smash{\begin{tabular}[t]{l}$\{3\}$\end{tabular}}}}%
    \put(0.23824234,0.41693862){\color[rgb]{0,0,0}\makebox(0,0)[lt]{\lineheight{1.25}\smash{\begin{tabular}[t]{l}$\{4\}$\end{tabular}}}}%
    \put(0,0){\includegraphics[width=\unitlength,page=2]{comp.pdf}}%
    \put(0.25587826,1.18062874){\color[rgb]{0,0,0}\makebox(0,0)[lt]{\lineheight{1.25}\smash{\begin{tabular}[t]{l}$\{1,2,3,4\}$\end{tabular}}}}%
    \put(0.57266132,1.01960615){\color[rgb]{0,0,0}\makebox(0,0)[lt]{\lineheight{1.25}\smash{\begin{tabular}[t]{l}$\{3,4\}$\end{tabular}}}}%
    \put(0.07726151,0.69884471){\color[rgb]{0,0,0}\makebox(0,0)[lt]{\lineheight{1.25}\smash{\begin{tabular}[t]{l}$\{3,4\}$\end{tabular}}}}%
    \put(0.0770948,1.00897849){\color[rgb]{0,0,0}\makebox(0,0)[lt]{\lineheight{1.25}\smash{\begin{tabular}[t]{l}$\{1,2\}$\end{tabular}}}}%
    \put(0.75102009,0.85523801){\color[rgb]{0,0,0}\makebox(0,0)[lt]{\lineheight{1.25}\smash{\begin{tabular}[t]{l}$\{3\}$\end{tabular}}}}%
    \put(0,0){\includegraphics[width=\unitlength,page=3]{comp.pdf}}%
    \put(-0.00356269,0.56214792){\color[rgb]{0,0,0}\makebox(0,0)[lt]{\lineheight{1.25}\smash{\begin{tabular}[t]{l}$\{4\}$\end{tabular}}}}%
    \put(0,0){\includegraphics[width=\unitlength,page=4]{comp.pdf}}%
    \put(-0.00011434,0.41849671){\color[rgb]{0,0,0}\makebox(0,0)[lt]{\lineheight{1.25}\smash{\begin{tabular}[t]{l}$\{4\}$\end{tabular}}}}%
    \put(0,0){\includegraphics[width=\unitlength,page=5]{comp.pdf}}%
    \put(0.0095692,0.27763778){\color[rgb]{0,0,0}\makebox(0,0)[lt]{\lineheight{1.25}\smash{\begin{tabular}[t]{l}$\{4\}$\end{tabular}}}}%
    \put(0,0){\includegraphics[width=\unitlength,page=6]{comp.pdf}}%
    \put(0.99408344,-2.36351645){\color[rgb]{0,0,0}\makebox(0,0)[lt]{\begin{minipage}{0.69805309\unitlength}\raggedright \end{minipage}}}%
    \put(2.21348062,-2.6897588){\color[rgb]{0,0,0}\makebox(0,0)[lt]{\begin{minipage}{1.09103972\unitlength}\raggedright \end{minipage}}}%
    \put(0.03140137,-1.79660362){\color[rgb]{0,0,0}\makebox(0,0)[lt]{\begin{minipage}{1.88257817\unitlength}\raggedright \end{minipage}}}%
    \put(0.74236412,0.70047661){\color[rgb]{0,0,0}\makebox(0,0)[lt]{\lineheight{1.25}\smash{\begin{tabular}[t]{l}$\{4\}$\end{tabular}}}}%
    \put(0,0){\includegraphics[width=\unitlength,page=7]{comp.pdf}}%
    \put(0.71667134,0.55254242){\color[rgb]{0,0,0}\makebox(0,0)[lt]{\lineheight{1.25}\smash{\begin{tabular}[t]{l}$\{4\}$\end{tabular}}}}%
    \put(0,0){\includegraphics[width=\unitlength,page=8]{comp.pdf}}%
    \put(0.7213379,0.41688549){\color[rgb]{0,0,0}\makebox(0,0)[lt]{\lineheight{1.25}\smash{\begin{tabular}[t]{l}$\{4\}$\end{tabular}}}}%
    \put(0,0){\includegraphics[width=\unitlength,page=9]{comp.pdf}}%
    \put(0.72926221,0.28839858){\color[rgb]{0,0,0}\makebox(0,0)[lt]{\lineheight{1.25}\smash{\begin{tabular}[t]{l}$\{4\}$\end{tabular}}}}%
    \put(0.23623341,0.28004869){\color[rgb]{0,0,0}\makebox(0,0)[lt]{\lineheight{1.25}\smash{\begin{tabular}[t]{l}$\{4\}$\end{tabular}}}}%
    \put(0,0){\includegraphics[width=\unitlength,page=10]{comp.pdf}}%
    \put(0.24090003,0.14439176){\color[rgb]{0,0,0}\makebox(0,0)[lt]{\lineheight{1.25}\smash{\begin{tabular}[t]{l}$\{4\}$\end{tabular}}}}%
    \put(0,0){\includegraphics[width=\unitlength,page=11]{comp.pdf}}%
    \put(0.24090003,0.01590486){\color[rgb]{0,0,0}\makebox(0,0)[lt]{\lineheight{1.25}\smash{\begin{tabular}[t]{l}$\{4\}$\end{tabular}}}}%
    \put(0.49835056,0.84869202){\color[rgb]{0,0,0}\makebox(0,0)[lt]{\lineheight{1.25}\smash{\begin{tabular}[t]{l}$\{4\}$\end{tabular}}}}%
    \put(0,0){\includegraphics[width=\unitlength,page=12]{comp.pdf}}%
    \put(0.51764743,0.7050408){\color[rgb]{0,0,0}\makebox(0,0)[lt]{\lineheight{1.25}\smash{\begin{tabular}[t]{l}$\{4\}$\end{tabular}}}}%
    \put(0,0){\includegraphics[width=\unitlength,page=13]{comp.pdf}}%
    \put(0.51148235,0.56418187){\color[rgb]{0,0,0}\makebox(0,0)[lt]{\lineheight{1.25}\smash{\begin{tabular}[t]{l}$\{4\}$\end{tabular}}}}%
    \put(0,0){\includegraphics[width=\unitlength,page=14]{comp.pdf}}%
  \end{picture}%
\endgroup%

%% file: comp2.pdf_tex
\begingroup%
  \makeatletter%
  \providecommand\color[2][]{%
    \errmessage{(Inkscape) Color is used for the text in Inkscape, but the package 'color.sty' is not loaded}%
    \renewcommand\color[2][]{}%
  }%
  \providecommand\transparent[1]{%
    \errmessage{(Inkscape) Transparency is used (non-zero) for the text in Inkscape, but the package 'transparent.sty' is not loaded}%
    \renewcommand\transparent[1]{}%
  }%
  \providecommand\rotatebox[2]{#2}%
  \newcommand*\fsize{\dimexpr\f@size pt\relax}%
  \newcommand*\lineheight[1]{\fontsize{\fsize}{#1\fsize}\selectfont}%
  \ifx\svgwidth\undefined%
    \setlength{\unitlength}{186.87059419bp}%
    \ifx\svgscale\undefined%
      \relax%
    \else%
      \setlength{\unitlength}{\unitlength * \real{\svgscale}}%
    \fi%
  \else%
    \setlength{\unitlength}{\svgwidth}%
  \fi%
  \global\let\svgwidth\undefined%
  \global\let\svgscale\undefined%
  \makeatother%
  \begin{picture}(1,0.7067356)%
    \lineheight{1}%
    \setlength\tabcolsep{0pt}%
    \put(0.25281109,0.19136964){\color[rgb]{0,0,0}\makebox(0,0)[lt]{\lineheight{1.25}\smash{\begin{tabular}[t]{l}$\{3\}$\end{tabular}}}}%
    \put(0,0){\includegraphics[width=\unitlength,page=1]{comp2.pdf}}%
    \put(0.27656512,0.65840309){\color[rgb]{0,0,0}\makebox(0,0)[lt]{\lineheight{1.25}\smash{\begin{tabular}[t]{l}$\{1,2,3,4\}$\end{tabular}}}}%
    \put(0.59745156,0.49529474){\color[rgb]{0,0,0}\makebox(0,0)[lt]{\lineheight{1.25}\smash{\begin{tabular}[t]{l}$\{3,4\}$\end{tabular}}}}%
    \put(0.08359429,0.33091729){\color[rgb]{0,0,0}\makebox(0,0)[lt]{\lineheight{1.25}\smash{\begin{tabular}[t]{l}$\{3,4\}$\end{tabular}}}}%
    \put(0.09546583,0.48452941){\color[rgb]{0,0,0}\makebox(0,0)[lt]{\lineheight{1.25}\smash{\begin{tabular}[t]{l}$\{1,2\}$\end{tabular}}}}%
    \put(0.77812066,0.32879748){\color[rgb]{0,0,0}\makebox(0,0)[lt]{\lineheight{1.25}\smash{\begin{tabular}[t]{l}$\{3\}$\end{tabular}}}}%
    \put(0,0){\includegraphics[width=\unitlength,page=2]{comp2.pdf}}%
    \put(0.00659665,0.19101643){\color[rgb]{0,0,0}\makebox(0,0)[lt]{\lineheight{1.25}\smash{\begin{tabular}[t]{l}$\{4\}$\end{tabular}}}}%
    \put(0,0){\includegraphics[width=\unitlength,page=3]{comp2.pdf}}%
    \put(-0.00360884,0.01892498){\color[rgb]{0,0,0}\makebox(0,0)[lt]{\lineheight{1.25}\smash{\begin{tabular}[t]{l}$\{4\}$\end{tabular}}}}%
    \put(0,0){\includegraphics[width=\unitlength,page=4]{comp2.pdf}}%
    \put(1.02433248,-2.93165038){\color[rgb]{0,0,0}\makebox(0,0)[lt]{\begin{minipage}{0.70709516\unitlength}\raggedright \end{minipage}}}%
    \put(2.25952484,-3.26211863){\color[rgb]{0,0,0}\makebox(0,0)[lt]{\begin{minipage}{1.10517226\unitlength}\raggedright \end{minipage}}}%
    \put(0.04918052,-2.35739417){\color[rgb]{0,0,0}\makebox(0,0)[lt]{\begin{minipage}{1.90696372\unitlength}\raggedright \end{minipage}}}%
    \put(0,0){\includegraphics[width=\unitlength,page=5]{comp2.pdf}}%
    \put(0.77812101,0.16255644){\color[rgb]{0,0,0}\makebox(0,0)[lt]{\lineheight{1.25}\smash{\begin{tabular}[t]{l}$\{4\}$\end{tabular}}}}%
    \put(0,0){\includegraphics[width=\unitlength,page=6]{comp2.pdf}}%
    \put(0.24705908,0.01611088){\color[rgb]{0,0,0}\makebox(0,0)[lt]{\lineheight{1.25}\smash{\begin{tabular}[t]{l}$\{4\}$\end{tabular}}}}%
    \put(0.53020518,0.3221667){\color[rgb]{0,0,0}\makebox(0,0)[lt]{\lineheight{1.25}\smash{\begin{tabular}[t]{l}$\{4\}$\end{tabular}}}}%
    \put(0,0){\includegraphics[width=\unitlength,page=7]{comp2.pdf}}%
    \put(0.54981398,0.16383573){\color[rgb]{0,0,0}\makebox(0,0)[lt]{\lineheight{1.25}\smash{\begin{tabular}[t]{l}$\{4\}$\end{tabular}}}}%
  \end{picture}%
\endgroup%

%% file: conclusion.tex
Our results can be considered as a step in the research direction recently proposed by Fernau, Wolf and Yamakami in~\cite{PushdownAutomataUndecidable},  in which the authors prove that 
the synchronisation problem for PDAs is undecidable when the stack is \emph{not} visible. They also suggest looking into different variants of 
synchronisation for PDAs with a view towards the decidability and complexity frontier.
Within this context, we believe we have proposed a natural variant of synchronisation
in which the observer can see the stack and given decidability and complexity-theoretic optimal results for both the non-deterministic and the deterministic cases.

One can ask similar questions to almost any automata model, which has a ``data'' part that
can be observed, for example, timed automata. Another natural question is to consider the same questions for one-counter automata, i.e., pushdown automata with a single stack alphabet. 
Though our results imply decidability for this case, further investigation into its complexity is needed. Finally, one can impose constraints on pushdown automata that are less constrained than determinism, but do not imply full non-determinism, like unambiguity or being Good-for-games.

%% file: appendix.tex
\section{Proofs for Section \ref{sec:equivalence}}

\begin{remark}~\label{remark:convention}
    For the sake of brevity, for some
    tuples $(q,a,A) \in Q \times \Sigma \times \Gamma$, we will sometimes \textbf{not} specify
    $\delta(q,a,A)$. In such cases, it is to be assumed that $\delta(q,a,A) = (q,A)$.
    Also, occasionally we will describe transitions by saying ``Upon reading $a$, the state $p$ moves to $q$'', without describing the changes to the stack. In these cases, it is to be assumed that irrespective of the element at the top of the stack, the state $p$ moves to $q$ and does not change the stack.
\end{remark}

\ThmEquivSuperSync*

We break the proof of this Lemma into various propositions, each one 
showing equivalence between a pair of problems.

\begin{proposition}
    {\scshape Subset-Ada-Sync} and {\scshape Given-Sync} are polynomial time 
    equivalent. Further the same applies for their corresponding deterministic versions.
\end{proposition}

\begin{proof}
    \paragraph*{{\scshape Subset-Ada-Sync} is reducible to {\scshape Given-Sync}}
    Let $\Pc = (Q,\Sigma,\Gamma,\delta)$ be a PDA with $I \subseteq Q$ and
    $\gamma \in \Gamma^*$. The central idea is that the new PDA $\Pc'$ that we construct will force the observer to initially decide on which state of $\Pc$ she wants to synchronise in, by inputting a special letter specific to each state of $\Pc$. Once she has done that, this choice will be remembered in the states of $\Pc'$. Finally, only when
    she believes she has synchronised in the state that she 
    had chosen initially, she can input a special letter which
    will take her to a special state $q_\acc$. 
    
    This can be concretely implemented as follows: The states of $\Pc'$
    will be $(Q \times Q) \cup (Q \times \{\smiley\}) \cup \{q_\acc,q_\rej\}$.
    The input alphabet of $\Pc'$ will be $\Sigma \cup \{\text{decide}_q : q \in Q\} \cup \{\text{done}_q : q \in Q\}$. The transition relation $\delta'$ of $\Pc'$
    is as follows:
    Upon reading $\text{decide}_q$,
    the state $(p,\smiley)$ moves to the state $(p,q)$.
    Upon reading any letter from $\Sigma$, the transitions on states $(p,q) \in Q \times Q$ just mimic the transitions of $\delta$ on the first co-ordinate
    and leave the second co-ordinate unchanged.
    Upon reading $\text{done}_q$, the states $(q,q)$ and $q_\acc$ move to $q_\acc$,
    while all the other states move to $q_\rej$. Notice 
    that $\Pc'$ is deterministic if $\Pc$ is. We now claim that
    \begin{quote}
        $\exists s \in Q$ with $(I,\gamma) \xRightarrow[\Pc]{} s$ if and only if $(I \times \{\smiley\},\gamma) \xRightarrow[\Pc']{} q_\acc$.    
    \end{quote}
    
    ($\Rightarrow$ ) Suppose there exists $s \in Q$ such that there
    is a synchroniser between $(I,\gamma)$ and $s$ in $\Pc$ (say $T$).
    Let $(S_v,\eta_v)$ be the label of each vertex $v$ in $T$.
    By converting the label $(S_v,\eta_v)$ to $(S_v \times \{s\},\eta_v)$ for 
    each vertex $v$, we get a synchroniser between $(I \times \{s\},\gamma)$ and $(s,s)$ in $\Pc'$. Now to the root of this synchroniser
    add a parent labelled by $(I \times \{\smiley\}, \gamma)$ with 
    its outgoing edge labelled by $\text{decide}_{s}$.
    Similarly, to each leaf $v$, add an outgoing edge labelled by $\text{done}_{s}$
    and label the child of this edge by $(q_\acc,\eta_v)$.
    By inspection, it can be easily verified that this new tree is 
    a synchroniser between $(I \times \{\smiley\},\gamma)$ and $q_\acc$ in $\Pc'$.\\
    
    ($\Leftarrow$ ) Suppose there exists a synchroniser between $(I \times \{\smiley\},\gamma)$ and $q_\acc$ in $\Pc'$ (say $T'$). Assume that $T'$ is a minimal
    such synchroniser.  We recall our convention stated in Remark~\ref{remark:convention} that if we did not 
    state the image of $\delta'(q,a,A)$ for some $q,a,A$, then
    $\delta'(q,a,A) = (q,A)$. 
    With this convention in mind, by using the 
    fact that $T'$ is minimal, we can easily conclude that the root has 
    exactly one outgoing edge labelled by $\text{decide}_{s}$ for some $s \in Q$ and the child of this edge is labelled by $(I \times \{s\}, \gamma)$.
    Hence, if we remove the root of $T'$, we get a synchroniser between
    $(I \times \{s\},\gamma)$ and $q_\acc$, which we will 
    denote by $T$. By our assumption on $T'$, it follows that $T$ is a minimal
    such synchroniser between $(I \times \{s\},\gamma)$ and $q_\acc$ in $\Pc'$.
    
    Let $(S_v,\eta_v)$ be the label of the vertex $v$ in the tree $T$.
    We now do a series of observations.
    \begin{itemize}
        \item No state in $Q \times \{\smiley\}$ has an incoming transition. Hence, 
        \begin{quote}\label{fact:A}
            Fact A: For every vertex $v$, $S_v \cap (Q \times \{\smiley\}) = \emptyset$.
        \end{quote}
        \item By our convention stated in remark~\ref{remark:convention}, it follows that $\delta'(p,\text{decide}_q,A) = (p,A)$
        for any $p \notin Q \times \{\smiley\}, A \in \Gamma$ and $q \in Q$.
        By fact~\ref{fact:A} and by the minimality of $T$ we get,
        \begin{quote}\label{fact:no-decide}
            Fact B: No edge of $T$ is labelled by $\text{decide}_q$ for any $q \in Q$.
        \end{quote}
        \item Notice that $q_\rej$ has no outgoing transitions. Since $S_u = \{q_{\acc}\}$
        for any leaf $u$, it follows that 
        \begin{quote}\label{fact:no-rej}
            Fact C: For every vertex $v$, $q_\rej \notin S_v$. 
        \end{quote}
        \item By induction on the structure of the tree $T$, we now prove that
        \begin{quote}\label{fact:D}
            Fact D: For all non-leaves $v$, $S_v \subseteq Q \times \{s\}$. Further, if the outgoing edge from some vertex $v$ is labelled
            by $\text{done}_q$ for some $q \in Q$, then $q = s$
            and $v$ is the parent of a leaf with $S_v = \{(s,s)\}$.
        \end{quote}
        Clearly $S_v \subseteq Q \times \{s\}$ when $v$ is the root vertex.
        Suppose for some non-leaf $v$, we have $S_v \subseteq Q \times \{s\}$. If the outgoing edges from $v$ are labelled
        by some letter from $\Sigma$, then it is clear that for all children
        $v'$ of $v$ we have $S_{v'} \subseteq Q \times \{s\}$.
        By Fact B no edge can be labelled by $\text{decide}_q$
        for any $q \in Q$. Hence, the only remaining case is when the outgoing
        edges of $v$ are labelled by $\text{done}_q$ for some $q \in Q$.
        In this case, there is only one child of $v$ (say $v'$).
        Since $S_v \subseteq Q \times \{s\}$, if $\text{done}_q \neq \text{done}_s$ or if $S_v \neq \{(s,s)\}$
        then $\rej \in S_{v'}$ which contradicts Fact C.
        Hence, $S_v = \{(s,s)\}$, $\text{done}_q = \text{done}_s$ and so $S_{v'} = \{q_\acc\}$. By minimality
        of $T$, $v'$ must be a leaf.
    \end{itemize}
    
    By combining all the facts it follows that if we remove the leaves
    of $T$, we get a synchroniser between $(I \times \{s\},\gamma)$
    and $(s,s)$ in $\Pc'$ such that all the edges are labelled by letters from $\Sigma$
    alone. Hence, if we remove the second co-ordinate $s$ from each
    state in the label of each vertex of $T$, we will get
    a synchroniser between $(I,\gamma)$ and $s$ in $\Pc$.\\

    Therefore, we have shown that  $\exists s \in Q$ with $(I,\gamma) \xRightarrow[\Pc]{} s$ if and only if $(I \times \{\smiley\},\gamma) \xRightarrow[\Pc']{} q_\acc$.  Hence, this gives the desired
    reduction from {\scshape Subset-Ada-Sync} to {\scshape Given-Sync}.

\paragraph*{{\scshape Given-Sync} is reducible to {\scshape Subset-Ada-Sync}}
    Let $\Pc = (Q,\Sigma,\Gamma,\delta)$ with $I \subseteq Q$, $s \in Q$
    and $\gamma \in \Gamma^*$. The central idea is to take 
    two disjoint copies of $\Pc$ and then add a new state $q_\acc$ such that
    the only state reachable from both these copies of $\Pc$ is the state $q_\acc$.
    Further this state $q_\acc$ can be reached only from the
    copies of the state $s$. Then, if at all synchronisation is possible 
    from both the copies of $I$ in $\Pc'$, it has to happen at $q_\acc$ and so must go through the corresponding copies of $s$. Hence the projection
    of this synchronisation on any of the copies will lead 
    to a synchronisation from $I$ to $s$ in $\Pc$. 
    
    This idea can be 
    concretely implemented as follows: The states of $\Pc'$
    will be $Q' = (Q \times \{0,1\}) \cup \{q_\acc,(q_\rej,0),(q_\rej,1)\}$.
    The input alphabet of $\Pc'$ will be $\Sigma \cup \{\End\}$.
    Upon reading any letter from $\Sigma$, the transitions on a state $(q,b) \in Q \times \{0,1\}$ just mimic the transitions of $\delta$ on the first co-ordinate
    and leave the second one unchanged. Upon reading $\End$, 
    the states $(s,0)$, $(s,1)$ and $q_\acc$ move to $q_\acc$ and 
    all the other states move to their corresponding copy of $q_\rej$.
    Notice 
    that $\Pc'$ is deterministic if $\Pc$ is. 
    We now claim that
    \begin{quote}
        $(I,\gamma) \xRightarrow[\Pc]{} s$ iff
        there exists $q \in Q'$ such that $(I \times \{0,1\}, \gamma) \xRightarrow[\Pc']{} q$.     
    \end{quote}
    
    ($\Rightarrow$) Suppose there is a synchroniser between $(I,\gamma)$
    and $s$ in $\Pc$ (say $T$). Let $(S_v,\eta_v)$ be the label
    of every vertex $v$ in $T$. Modify $T$ as follows: Replace $(S_v,\eta_v)$
    with $(S_v \times \{0,1\},\eta_v)$ for every vertex $v$.
    Then to each leaf $v$, add an outgoing edge labelled by $\End$ and
    label the child of this edge by $(q_\acc,\eta_v)$.  
    By inspecting the transition relation, it can be easily verified
    that this modified tree is a synchroniser between $(I \times \{0,1\}, \gamma)$
    and $q_\acc$ in $\Pc'$.\\
    
    ($\Leftarrow$) Suppose there exists $q \in Q'$ such that $(I \times \{0,1\}, \gamma) \xRightarrow[\Pc']{} q$.  Since $q_\acc$ is the only state
    reachable from both $I \times \{0\}$ and $I \times \{1\}$, it
    follows that $q = q_\acc$. Hence, we have a synchroniser between
    $(I \times \{0,1\},\gamma)$ and $q_\acc$ in $\Pc'$ (say $T$). 
    We can assume that $T$ is a minimal such synchroniser.
    Let $(S_v,\eta_v)$ be the label of each vertex $v$ in $T$.
    We now do a series of observations.
    \begin{itemize}
        \item Notice that $(q_\rej,0)$ and $(q_\rej,1)$ both have no outgoing transitions. Since $S_u = \{q_{\acc}\}$ for any leaf $u$, it follows that 
        \begin{quote}
            Fact A: For every vertex $v$, $(q_\rej,0) \notin S_v$ and $(q_\rej,1) \notin S_v$.
        \end{quote}
        
        \item By induction on the structure of the tree $T$, we now prove that
        \begin{quote}
            Fact B: For all non-leaf vertices $v$, there exists $Q_v \subseteq Q$
            such that $S_v = Q_v \times \{0,1\}$. 
            Further, if the outgoing edge from some vertex $v$ is labelled
            by $\End$, then $v$ is the parent of a leaf with $S_v = \{s\} \times \{0,1\}$.
        \end{quote}
        
        Clearly $S_v = I \times \{0,1\}$ when $v$ is the root.
        Suppose for some non-leaf vertex $v$, there exists $Q_v$ such that $S_v = Q_v \times \{0,1\}$. If the outgoing edges from $v$ are labelled
        by some letter from $\Sigma$, then it is clear that for all children
        $v'$ of $v$ there exists $Q_{v'}$ with $S_{v'} = Q_{v'} \times \{0,1\}$.
        Hence, the only remaining case is when the outgoing
        edges of $v$ are labelled by $\End$.
        In this case, there is only one child of $v$ (say $v'$).
        If $Q_v \neq \{s\}$ then $(q_\rej,0), (q_\rej,1) \in S_{v'}$ which contradicts Fact A.
        Hence, $Q_v = \{s\}$ and so $S_v = \{s\} \times \{0,1\}$. Therefore, it follows that $S_{v'} = \{q_\acc\}$. By minimality
        of $T$, $v'$ must be a leaf.
    \end{itemize}
    
    It then follows that if we  remove the  leaves of $T$, we get
    a synchroniser between $(I \times \{0,1\},\gamma)$ and $s \times \{0,1\}$
    such that all the edges are labelled by letters from $\Sigma$ alone. 
    Hence, if we project the labels of each vertex on the first copy,
    we get a synchroniser between $(I,\gamma)$ and $s$ in $\Pc$.\\
    
    Therefore we have shown that $(I,\gamma) \xRightarrow[\Pc]{} s$ iff
    there exists $q \in Q'$ such that $(I \times \{0,1\}, \gamma) \xRightarrow[\Pc']{} q$. This gives the desired reduction from {\scshape Given-Sync} to
    {\scshape Subset-Ada-Sync}.
    
\end{proof}

\begin{proposition}
    {\scshape Given-Sync} is poly. time equivalent to {\scshape Super-Sync}. 
    Further, the same is true for the corresponding deterministic versions.
\end{proposition}

\begin{proof}
    \paragraph*{{\scshape Given-Sync} is reducible to {\scshape Super-Sync}}
    Let $\Pc = (Q,\Sigma,\Gamma,\delta)$ be a PDA with 
    $I \subseteq Q$, $s \in Q$ and $\gamma \in \Gamma^*$.
    Construct $\Pc'$ from $\Pc$ by adding two new states $q_\acc$ and $q_\rej$
    and two new input letters $\End$ and $\mathtt{pop}$. Upon inputting $\End$
    the states $s$ and $q_\acc$ move to $q_\acc$ whereas all the
    other states move to $q_\rej$. Upon inputting $\mathtt{pop}$,
    all the states except $q_\acc$ move to $q_\rej$ whereas
    $q_\acc$ remains at $q_\acc$ and keeps on popping the stack.
    Notice that $\Pc'$ is deterministic if $\Pc$ is.
    We now claim that 
    \begin{quote}
        $(I,\gamma) \xRightarrow[\Pc]{} s$ if and only if $(I,\gamma) \xRightarrow[\Pc']{\text{sup}} q_\acc$
    \end{quote}
    
    ($\Rightarrow$) Let $T$ be  a synchroniser
    between $(I,\gamma)$ and $s$ in $\Pc$ and let
    $(S_v,\eta_v)$ be the label of each vertex $v$. 
    From each leaf $v$, add an outgoing edge  labelled by $\End$ and label
    the child of this edge by $(q_\acc,\eta_v)$. 
    Let $\eta_v = w_0w_1 \dots w_k\bot$.
    Now, add a chain of $k+1$ vertices from this child
    with each edge labelled by $\mathtt{pop}$ such that the $i^{th}$ vertex
    in the chain is labelled by $(q_\acc,w_iw_{i+1},\dots,w_k,\bot)$.
    It is clear that the new tree is a super-synchroniser between
    $(I,\gamma)$ and $q_\acc$ in $\Pc'$.\\
    
    ($\Leftarrow$) Let $T$ be a super-synchroniser between $(I,\gamma)$ and
    $q_\acc$ in $\Pc'$. We can assume that $T$ is a minimal such super-synchroniser.
    Let $(S_v,\eta_v)$ be the label of each vertex $v$ in $T$.
    We now do a series of observations.
    \begin{itemize}
        \item Because there are no outgoing transitions from $q_\rej$ and since $S_u = \{q_\acc\}$ for every leaf $u$, we have,
        \begin{quote}
            Fact A: For every vertex $v$, $q_\rej \notin S_v$.
        \end{quote}
        \item We now claim that,
        \begin{quote}
            Fact B: Along every branch of $T$, there is a vertex $v$ with only one child $v'$
            such that $S_v = \{s\}$, the outgoing edge from $v$ is labelled
            by $\End$ and $S_{v'} = \{q_\acc\}$. Further, for every vertex $v''$ before $v$
            in this branch, we have $S_{v''} \subseteq Q$ and no edge before the edge $(v,v')$
            along this branch is labelled by $\End$.
        \end{quote}
        
        Let us consider a branch of the tree $T$. We will
        say that $v \le v'$ for two vertices along this branch if $v$ appears
        before $v'$ along this branch.
        Now, the root of the branch is labelled by $(I,\gamma)$ whereas the leaf is 
        labelled by $(q_\acc,\bot)$. Hence, there should be a vertex
        $v$ and its child $v'$ along this branch such that $q_\acc \in S_{v'}$
        but $q_\acc \notin S_{v''}$ for every vertex $v'' \le v$.
        By Fact A, $q_\rej \notin S_{v''}$ for every vertex $v'' \le v$ and so $S_{v''} \subseteq Q$ for every $v'' \le v$. Further, if 
        the outgoing edge from some vertex $v'' < v$ 
        is labelled by $\End$, then the child of $v''$ along this 
        branch will contain either $q_\acc$ or $q_\rej$, which will lead to 
        a contradiction. 
        
        Now the only way to move from some state in $Q$ to $q_\acc$ is by the letter $\End$.
        Hence the edge between $v$ and $v'$ must be labelled by $\End$. 
        Now if $S_v \neq \{s\}$ then $q_\rej \in S_{v'}$ which contradicts Fact A.
        Hence $S_v = \{s\}$ and $S_{v'} = \{q_\acc\}$.
        
    \end{itemize}
    
    Hence, using Fact B, we proceed to cut the tree $T$ as follows:
    Along every branch, find the vertex $v$ as guaranteed
    by Fact B and then remove all the vertices after $v$ along this branch.
    It follows that this reduced tree will be a synchroniser between
    $(I,\gamma)$ and $s$ in $\Pc$.\\
    
    Therefore we have shown that $(I,\gamma) \xRightarrow[\Pc]{} s$ if and only if $(I,\gamma) \xRightarrow[\Pc']{\text{sup}} q_\acc$. Hence
    we get the desired reduction from {\scshape Given-Sync} to {\scshape Super-Sync}.
    
    \paragraph*{{\scshape Super-Sync} is reducible to {\scshape Given-Sync}}
    Let $\Pc = (Q,\Sigma,\Gamma,\delta)$ be a PDA with 
    $I \subseteq Q$, $s \in Q$ and $\gamma \in \Gamma^*$.
    Construct $\Pc'$ from $\Pc$ by adding two new states $q_\acc$ and $q_\rej$
    and one new input letter $\End$. Upon inputting $\End$,
    the state $q_\acc$ remains at $q_\acc$, 
    the state $s$ moves to $q_\acc$ if the stack is empty
    and in all the other cases, $\Pc'$ moves to $q_\rej$.
    Notice that $\Pc'$ is deterministic if $\Pc$ is.
    We now claim that 
    \begin{quote}
        $(I,\gamma) \xRightarrow[\Pc]{\text{sup}} s$ if and only if $(I,\gamma) \xRightarrow[\Pc'] {} q_\acc$
    \end{quote}
    
    ($\Rightarrow$) Suppose $T$ is a super-synchroniser between $(I,\gamma)$
    and $s$ in $\Pc$. Let $(S_v,\eta_v)$ be the label of each vertex $v$.
    To every leaf of $T$, add an outgoing edge labelled by $\End$,
    and label the child of this edge by $(q_\acc,\bot)$. 
    It follows that this new tree is a synchroniser between $(I,\gamma)$
    and $q_\acc$ in $\Pc'$.\\
    
    ($\Leftarrow$) Suppose $T$ is a synchroniser between $(I,\gamma)$
    and $q_\acc$ in $\Pc'$. We can assume $T$ is a minimal such synchroniser.
    Let $(S_v,\eta_v)$ be the label of each vertex $v$.
    We now do a series of observations.
    \begin{itemize}
        \item Since there are no outgoing transitions from $q_\rej$, and since $S_u = \{q_\acc\}$ for every leaf $u$, it follows that
        \begin{quote}
            Fact A: For every vertex $v$, $q_\rej \notin S_v$.
        \end{quote}
        \item By induction on the structure of $T$, we claim that,
        \begin{quote}
            Fact B: If $v$ is a non-leaf, then $S_v \subseteq Q$.
            Further, if the outgoing edge from some vertex $v$ is labelled by
            $\End$, then $v$ is the parent of a leaf with $(S_v,\eta_v) = (s,\bot)$.
        \end{quote}
        
        Clearly $S_v \subseteq Q$ when $v$ is the root.
        Suppose for some non-leaf $v$, $S_v \subseteq Q$. 
        If the outgoing edges from $v$ are labelled
        by some letter from $\Sigma$, then it is clear that for all children
        $v'$ of $v$, $S_{v'} \subseteq Q$.
        Hence, the only remaining case is when the outgoing
        edges of $v$ are labelled by $\End$.
        In this case, there is only one child of $v$ (say $v'$).
        If $S_v \neq \{s\}$ or if $\eta_v \neq \bot$, then $q_\rej \in S_{v'}$ which contradicts Fact A.
        Hence, $S_v = \{s\}$, $\eta_v = \bot$ and so $S_{v'} = \{q_\acc\}$. By minimality
        of $T$, $v'$ must be a leaf.
        
    \end{itemize}
    Hence, if we remove all the leaves of $T$ we get a 
    super-synchroniser between $(I,\gamma)$ and $s$ in $\Pc$.\\
    
    Therefore, we have shown that $(I,\gamma) \xRightarrow[\Pc]{\text{sup}} s$ iff $(I,\gamma) \xRightarrow[\Pc'] {} q_\acc$. This gives the desired
    reduction from {\scshape Super-Sync} to {\scshape Given-Sync}.
\end{proof}

\begin{proposition}
    {\scshape Super-Sync} and {\scshape Special-Sync} are poly. time equivalent. Further the same is true
    for the corresponding deterministic versions.
\end{proposition}

\begin{proof}
    It suffices to show that {\scshape Super-Sync} is reducible to
    {\scshape Special-Sync} as the latter is a special case of the former.
    Let $\Pc = (Q,\Sigma,\Gamma,\delta)$ be a PDA with $I \subseteq Q$,
    $s \in Q$ and $\gamma \in \Gamma^*$. Construct $\Pc'$ from $\Pc$
    by adding a new set of states $I' = \{q' : q \in I\}$ such that
    upon inputting any letter, the state $q' \in I'$
    pushes $\gamma$ onto the stack and moves to $q$. 
    It is obvious that $(I,\gamma) \xRightarrow[\Pc]{\text{sup}} s$
    iff $(I',\bot) \xRightarrow[\Pc']{\text{sup}} s$.
\end{proof}

\section{Proofs of Section~\ref{sec:lower-bounds}}
\input{Appendix1_ABP-NBP-lowerBound.tex}

\subsection{Proof of Reduction From Alternating Extended Pushdown Systems to {\scshape Special-Sync}}
\input{Appendix2-lower-bound-reduction.tex}

\section{Proofs of Section~\ref{sec:upperbounds}}

Throughout this section, we fix a single PDA $\Pc = (Q,\Sigma,\Gamma,\delta)$
with $I \subseteq Q$ and $s \in Q$. This gives
rise to the alternating pushdown system $\Ac_{\Pc} = (2^Q,\Gamma,\Delta,I,\{s\})$.

\PropEquiClassesSmallSize*

\begin{proof}
    By definition $T^a_{S,A} = \{t \in \delta : t = (p,a,A,q,\gamma) \text{ where } p \in S \}$. Because $\Pc$ is deterministic, the size of 
    $T^a_{S,A}$ is $|S|$. Now, the relation $\sim^a_{S,A}$ partitions
    $T^a_{S,A}$ into equivalence classes $E_1,\dots,E_k$ 
    and for each $i$, $\nex(E_i)$ is simply the set $\{q : (p,a,A,q,\gamma) \in E_i\}$. Since $|T^a_{S,A}|$ is $|S|$, it follows that $\sum_{i=1}^k |\nex(E_i)|$
    is at most $|S|$.
\end{proof}

\LemSmallLeaves*

\begin{proof}
    Let $T$ be any accepting run of $(S,\gamma) = (S,A\eta)$.
    We proceed by induction on the size of $T$.
    The base case of $1$ is trivial.
    For the induction step, suppose the size of $T$ is $m+1$ for some $m \ge 0$.
    Let $v_1,\dots,v_k$ be the children of the root.
    By nature of the transitions in $\Ac_{\Pc}$, it follows
    that there exists $a \in \Sigma$ and equivalence 
    classes $E_1,\dots,E_k$ of $\sim^a_{S,A}$ such that
    $v_i$ is labelled by 
    $(\nex(E_i),\word(E_i)\eta)$. 
    By induction hypothesis, the sub-tree rooted at $v_i$ has
    at most $|\nex(E_i)|$ leaves.
    By proposition~\ref{prop:equi-classes-small-size}
    we have that $\sum_{i=1}^k |\nex(E_i)| \le |S|$.
    Hence the total number of leaves of the tree $T$ is at most $|S|$.
\end{proof}

\subsection{Proofs for subsection~\ref{subsec:Sparse-Empty}}

Let us fix an alternating pushdown system $\Ac = (Q,\Gamma,\Delta,\init,\fin)$
and a number $k$. From $\Ac$, we can derive a non-deterministic pushdown system
obtained by deleting all transitions of the form $(q,A) \hookrightarrow \{(q_1, \gamma_1),\dots,(q_k,\gamma_k)\}$ with $k > 1$. 
We will denote this NPS by $\Nc$.

\LemCompress*

\begin{proof}
    ($\Rightarrow$) : Suppose we have a $k$-accepting run of $\Ac$ from $(p,\eta)$, say $T$.
    Let us proceed by induction on $|T|$.
    If $|T| = 1$, we are done. Otherwise, let $r$ be the root of $T$.
    If $r$ is a simple vertex, then let $v$ be the unique closest descendant 
    of $r$ such that $v$ is complex (Such a vertex always exists by means 
    of the definition of simple and complex). Note that the
    sub-tree rooted at $v$ has also $k$ leaves. If we let $(p',\eta')$
    be the label of $v$, by induction hypothesis there is a $k$-compressed accepting run from $(p',\eta')$, say $T'$. Now, take $T'$, 
    and add $(p,\eta)$ as a parent to $(p',\eta')$ in $T'$. By definition,
    this then gives rise to a $k$-compressed accepting run from $(p,\eta)$.
    
    If $r$ is a complex vertex, let $v_1,\dots,v_m$ be the children of $r$
    such that the label of each $v_i$ is $(p_i,\eta_i)$ and the sub-tree rooted at $v_i$ has $\ell_i$ leaves.
    By induction hypothesis, for each $i$, there is a $\ell_i$-compressed
    accepting run from $(p_i,\eta_i)$. Taking all these trees and adding
    $(p,\eta)$ as their root, gives rise to a $k$-compressed accepting run from $(p,\eta)$.\\
    
    ($\Leftarrow$) : Suppose we have a $k$-compressed accepting run of $\Ac$ from $(p,\eta)$, say $T$. Let us proceed by induction on $|T|$.
    If $|T| = 1$, we are done. Otherwise, let $r$ be the root of $T$.
    If $r$ is a simple vertex, then let $v$ be the only child of 
    $r$ and let the label of $v$ be $(p',\eta')$. 
    By definition of the $k$-compressed accepting run $T$,
    we have a run $(p,\eta) \xrightarrow[\Nc]{} (p_1,\eta_1) \xrightarrow[\Nc]{} (p_2,\eta_2) \dots (p_m,\eta_m) \xrightarrow[\Nc]{}(p',\eta')$.
    By induction hypothesis, we have a $k$-accepting run of $\Ac$
    from $(p',\eta')$, say $T'$. 
    Now, take $T'$ and attach  the linear chain of vertices $(p,\eta),(p_1,\eta_1),\dots,(p_m,\eta_m)$ before its root.
    This gives rise to a $k$-accepting run of $\Ac$ from $(p,\eta)$.
    
    If $r$ is a complex vertex, let $v_1,\dots,v_m$ be the children of $r$
    such that the label of each $v_i$ is $(p_i,\eta_i)$ and the sub-tree rooted at $v_i$ has $\ell_i$ leaves.
    By induction hypothesis, for each $i$, there is a $\ell_i$-accepting run from $(p_i,\eta_i)$. Taking all these trees and adding
    $(p,\eta)$ as their root, gives rise to a $k$-accepting run from $(p,\eta)$.
\end{proof}

\PropInvariant*

\begin{proof}
    Recall that Invariant (*) was the following:
    \begin{quote}
    Invariant (*) : A configuration $(q,\gamma) \in \Cc(M_v)$ iff all the vertices of the sub-tree rooted at $v$ can be labelled such that
    the resulting labelled sub-tree is a compressed accepting run of $\Ac$ from $(q,\gamma)$.
    \end{quote}
    
    Let us proceed by induction on the structure of the tree $T$.
    By construction, the invariant is true for all leaves $v$.
    Now, suppose we have a simple vertex $v$. Let $u$ be its only child.
    By induction hypothesis, assume that the invariant is true for $u$.
    By construction, $M_v$ is an automaton
    such that $\Cc(M_v) = \{(q',\gamma') : \exists (q,\gamma) \in \Cc(M_{u}) \text{ such that } (q',\gamma') \xrightarrow[\Nc]{*} (q,\gamma)\}$.
    It then immediately follows that the invariant is satisfied for the vertex $v$ as well.
    
    Suppose we have a complex vertex $v$ and let $v_1,\dots,v_\ell$ be its children.
    Suppose $(p,A\gamma) \in \Cc(M_v)$. By construction of 
    $M_v$, it then follows that there exists a transition
    $(p,A) \hookrightarrow \{(p_1,\gamma_1),\dots,(p_\ell,\gamma_\ell)\}$ of $\Ac$
    such that for each $i$, the configuration $(p_i,\gamma_i\gamma) \in \Cc(M_{v_i})$. By induction hypothesis, for each $i$, the sub-tree
    rooted at $v_i$ can be labelled so that the resulting labelled
    sub-tree is a compressed accepting run from $(p_i,\gamma_i\gamma)$.
    By taking this labelling for each of the sub-trees rooted at $v_1,\dots,v_\ell$ and then labelling the vertex $v$ by $(p,A\gamma)$,  we get 
    a labelling of the sub-tree rooted at $v$ which is a compressed accepting
    run from the configuration $(p,A\gamma)$.
    
    Conversely, suppose for some configuration $(p,A\gamma)$, it is
    possible to label the sub-tree rooted at $v$ so that it becomes
    a compressed accepting run from the configuration $(p,A\gamma)$.
    Hence, there exists a transition $(p,A) \hookrightarrow \{(p_1,\gamma_1),\dots,(p_\ell,\gamma_\ell)\}$ of $\Ac$
    such that for each $i$, the label of $v_i$ under this labelling 
    is $(p_i,\gamma_i\gamma)$. By induction hypothesis, for each $i$,
    we have that $(p_i,\gamma_i\gamma) \in \Cc(M_{v_i})$. 
    By construction of $M_v$, it follows that $(p,A\gamma) \in \Cc(M_v)$.
    Hence, the invariant is satisfied when $v$ is a complex vertex as well.
\end{proof}

\paragraph*{Running time analysis}

Let us analyse the running time of $\mathtt{Check}$.
Let $T$ be a $k$-structured tree and therefore $T$ has $O(k^2)$ vertices.
$\mathtt{Check}$ assigns to each vertex $v$ of $T$ an automaton $M_v$.
We claim that the running time of $\mathtt{Check}$ is $O(k^2 \cdot |\Ac|^{ck^2})$ (for some fixed constant $c$) because
of the following facts: 
\begin{itemize}
    \item[1)] By induction on the structure of the tree $T$, it can be proved that, there exists a constant $d$, such that if $h_v$ is the height of 
a vertex $v$ and $l_v$ is the number of leaves in the sub-tree of $v$, then the number of states of $M_v$ is $O(|\Ac|^{dh_vl_v})$ (Recall that $h_vl_v$ is at most $O(k^2$)). 
    \item[2)]  If an $\Nc$-automaton has $n$ states,
then the number of transitions it can have is $O(|\Ac| \cdot n^2)$. 
    \item[3)] For a vertex $v$
with children $v_1,\dots,v_\ell$, $M_v$ can be constructed in polynomial time
in the size of $|M_{v_1}| \times |M_{v_2}| \times \dots |M_{v_\ell}|$ and
$|\Ac|$.

\end{itemize}
 
Notice that everything else apart from Fact 1) is easy to see.
To prove Fact 1), we proceed by bottom-up induction on the structure of the tree $T$.
For the base case when the vertex $v$ is a leaf, notice that we can easily construct
the required automaton $M_v$ with at most $O(|\Ac|)$ states.  
Suppose, $v$ is a simple vertex and $u$ its only child. 
By Theorem~\ref{th:prestar}, $M_v$ has the same set of states as $M_u$.
By induction hypothesis, the number of states of $M_u$ is $O\left(|\Ac|^{dh_ul_u}\right)$ 
and so the number of states of $M_v$ is $O\left(|\Ac|^{dh_vl_v}\right)$.
Suppose $v$ is a complex vertex and $v_1,\dots,v_\ell$ are its children.
Let $h$ be the maximum height amongst the vertices $v_1,\dots,v_\ell$.
By induction hypothesis, the number of states of each $M_{v_i}$ is 
$O\left(|\Ac|^{dhl_{v_i}}\right)$. It is then clear that 
the number of states of $M_v$ is $O\left(\prod_{i=1}^\ell |\Ac|^{dhl_{v_i}} + |\Ac|\right)
= O\left(|\Ac|^{dhl_v} + |\Ac|\right) = O\left(|\Ac|^{d(h+1)l_v}\right) = O\left(|\Ac|^{dh_vl_v}\right)$.

Now the final algorithm for {\scshape Sparse-Empty} simply iterates over all $k$-structured trees and calls $\mathtt{Check}$ on all of them. Since the
number of $k$-structured trees is at most $f(k)$ where $f$ is an exponential function,
it follows that the total running time is $O\left(f(k) \cdot k^2 \cdot |\Ac|^{ck^2}\right) = O(|\Ac|^{ek^2})$ for some constant $e$.
\section{Homing Problem}
\input{Appendix3_Homing.tex}

%% file: Appendix1_ABP-NBP-lowerBound.tex
\LemABPTwoEXP*

\begin{proof}

We show that the acceptance problem for alternating Turing machines with exponential space can be reduced to the emptiness problem for AEPS.
Since alternating exponential space machines correspond
to deterministic doubly-exponential time machines, it would then follow that the emptiness problem is $\TWOEXPTIME$-hard.

More specifically, we are given an one-tape alternating Turing machine $\Mc$, a word $w$ and a number $B$ encoded in binary, and the problem is to decide if $M$ accepts $w$ whilst using at most $B$ tape cells. The reduction that we present here is similar to the reductions given in Theorem 5.4 of~\cite{Alternation} (and also Prop. 31 of~\cite{PushdownGames}), to prove that the emptiness problem for  AEPS without any Boolean variables is $\EXPTIME$-hard.
The only additional insight that we have here is that by using the Boolean variables in an AEPS, one can push \emph{exponentially} many symbols onto the stack in a single path before cycling back to some state. This is because using the tests and commands of an AEPS, once can implement a 'counter' which can count up to some exponential value. If $V$ denotes the set of variables in an AEPS, then we can store a number between 0 and $2^{|V|}-1$ using the
following convention: If $x_1,\dots,x_{\ell}$ are the
values of the variables $v_1,\dots,v_{\ell}$ at some point, then, at that point the values of the variables $V$ denote the number whose binary representation is given by $\boldsymbol{x} := x_1\cdot x_2\cdots x_\ell$  
where $x_1$ is the most significant bit and $x_\ell$, the least significant bit. Using some tests and commands, it is easy to see that one can implement operations which would effectively perform addition by 1, or checking for equality to a specific value, say $2^{|V|}$.

Let $Q$ be the states of $\Mc$, $\Sigma$ be the tape alphabet. A configuration of $\Mc$ will be denoted by the string $wqw'$ where $w,w'\in\Sigma^*$, $q \in Q$ and where the head of the machine is always to the right of the control state.
Let $\delta$ be the transition relation of $\Mc$ where transitions are of the form:
$$(q,a) \rightarrow \{(q_1,a_1,d_1),\dots,(q_k,a_k,d_k)\}$$
where $q,q_1,\dots,q_k \in Q$, $a,a_1,\dots,a_k \in \Sigma$ and $d_1,\dots,d_k \in \{\text{left},\text{right}\}$.
We note that the existential branching of the alternation is captured by non-deterministically choosing a transition applicable at each configuration and the universal branching is captured by forking into many copies as specified by the chosen transition.

We now construct an AEPS $\Ac$ which will guess and verify an accepting run of the machine $\Mc$ on the input $w$. $\Ac$ will operate in two stages. In the first stage, it guesses an accepting run of $\Mc$. In the second stage it verifies that this guess is indeed a valid run of the machine $\Mc$. 

The machine $\Ac$ will have $\ell = \lceil\log_2(B+2)\rceil$ many Boolean variables.
As mentioned before, using these Boolean variables we can have a bounded-counter
which will enable us to count up till $B+1$ and also allow us to check
if the counter value at any point is equal to some specific value (say something like $B+1$ or $B+2$). We will, in the description of the machine, use phrases like, `checks if a value is $x$', `pushes $y$ onto the stack $x$ many times', to denote counting using the Boolean variables $V$.

\paragraph*{The first stage}
Using the Boolean variables and appropriately designed tests and commands, $\Ac$ will have transitions which will allow it to push, in a non-deterministic manner, exactly $B+1$ letters from the set $Q \cup \Sigma$ onto the stack.
Additionally $\Ac$ also ensures, using the finite control, that the word pushed is of the form $w\cdot q\cdot w'$ where $w,w'\in\Sigma^*$ and $q\in Q$.
Once such a word has been pushed into the stack, the variables $V$ are reset to $0$ and we note that at this point, $\Ac$ has pushed a configuration of the machine $\Mc$ onto the stack. Note that in its finite control, $\Ac$ remembers the state $q$ that it pushed into the stack
and the letter $a$ that it pushed after $q$. Now, $\Ac$ non-deterministically picks a transition of $\Mc$ of the form $(q,a) \rightarrow \{(q_1,a_1,d_1),\dots,(q_k,a_k,d_k)\}$ and then forks into $k$ copies of itself, with the $i^{th}$ copy pushing the symbol $(q_i,a_i,d_i)$ into the stack. After this, it repeats this whole process again of trying to push a configuration and a transition of $\Mc$ onto the stack. The first stage ends when $\Ac$ pushes an accepting configuration of $\Mc$ onto the stack, i.e., a configuration where the state is an accepting state of $\Mc$.

Note that at this point, in each of the forked copies of $\Ac$, the stack has a sequence of the form
$$c_0 (q_1, a_1, d_1) c_1 (q_2, a_2, d_2) \dots (q_{k}, a_{k}, d_{k}) c_k$$
where each $c_i$ is a configuration of $\Mc$ and $c_k$ is an accepting configuration of $\Mc$

\paragraph*{The second stage}
$\Ac$ verifies that the guessed sequence is indeed a valid run of $\Mc$.
If $c_i$ is the configuration at the top of the stack, $\Ac$ forks into two copies,
with the first copy deciding to verify that the configuration $c_i$ follows from the configuration $c_{i-1}$ using the move $(q_i,a_i,d_i)$ and the second copy deciding to pop the configuration $c_i$ and $(q_i,a_i,d_i)$ from the stack and recursively doing a similar fork to verify the run from the configuration $c_{i-1}$.

To verify that $c_i$ follows from $c_{i-1}$ using $(q_i,a_i,d_i)$, the first copy of $\Ac$ proceeds as follows:
It forks into two copies, with the first copy deciding to check that the current letter of $c_i$ at the top of the stack follows correctly from the configuration $c_{i-1}$ using $(q_i,a_i,d_i)$ and the second copy popping the letter at the top of the stack and then recursively forking to do a similar choice for the next letter of $c_i$.
The first copy remembers the letter at the top of the stack, pops this letter and then using the bounded-counter removes $B+1$ symbols from the stack, and on the way to popping these $B+1$ symbols, also remembers the move
$(q_i,a_i,d_i)$ that it pops. After having removed these $B+1$ symbols, it then pops the next four letters from the stack, remembers all these four letters, and using the six pieces of information in the finite control that it has remembered, checks the consistency of these four letters along with the letter from the configuration $c_i$. If this check succeeds, then $\Ac$ moves to an accepting state, else it moves to a rejecting state.

Finally, when a copy of $\Ac$ ends up popping everything on the stack except for the configuration $c_0$,
it checks if this configuration is an initial configuration of the machine $\Mc$, i.e., it checks if 
it is of the form $qw\#^{B-n}$ where $\#$ denotes the blank symbol of $\Mc$ and $n$ is the size of the input $w$. To do this, it keeps on popping
the stack till a non-blank symbol is reached, and then using its finite control, checks that the remaining portion in the stack is of the form $qw$. 

It is clear from the description that such an alternating extended pushdown system $\Ac$ can be constructed in polynomial time.
\end{proof}

\LemmaNBPEXP*

\begin{proof}
    We will give a reduction from the problem of checking if an alternating linearly bounded Turing machine $\Mc$ accepts a word $w$, i.e., whether an alternating Turing machine
    $\Mc$ accepts a word $w$ whilst using at most $n = |w|$ tape cells.

    The proof of this can be seen as an adaptation of the proof of Lemma~\ref{lemma:ABP-2EXP} for alternating Turing machines that use linear space instead of exponential space. But now, since there is no alternation at the disposal of the machine, it instead simulates all possible branches of the alternating linear-space Turing machine. 

    We assume that the Turing machine $\Mc$ has a finite set of states $Q$ and tape alphabet $\Sigma$. A configuration of the machine is represented by a word in $\Sigma\uplus Q$ and is of the form $wqw'$ where $w,w'\in\Sigma^*$, $q\in Q$  and the head of the machine is always to the right of the control state. The transitions of $\Mc$ are of the form $(q,a) \rightarrow \{(q_1,a_1,d_1),\dots,(q_k,a_k,d_k)\}$. As discussed before, the existential branching of the alternation is captured by non-deterministically choosing a transition applicable at each configuration and the universal branching is captured by forking into many copies as specified by the chosen transition.

    Note that a run of an alternating linearly bounded Turing machine can be represented by a tree where the nodes are labelled by configurations. We now construct
    an NEPS $\Ac$, which, roughly speaking,  will explore this tree in a DFS order.
    Unlike in the proof of Lemma~\ref{lemma:ABP-2EXP}, here instead of non-deterministically pushing a configuration and later verifying it, with the help of polynomially many counters, $\Ac$ 'remembers' a configuration and ensures that the next configuration pushed respects the transition that is chosen to be executed.

    The machine $\Ac$ has $\lceil\log(|\Sigma|+|Q|)\rceil(n + 1)$ many Boolean variables. These variables are used to encode a configuration of the Turing machine as follows: The first $\lceil\log(|\Sigma|+|Q|)\rceil$ letters are used for the first letter of the configuration, the next $\lceil\log(|\Sigma|+|Q|)\rceil$ for the next and so on. 
    The machine $\Ac$ has three modes: The initial, forward and reverse mode.

\begin{itemize}
    \item In the \textbf{initial mode}, $\Ac$ pushes the initial configuration $c_0 = q_0\cdot w$ into the stack, one by one letter at a time
    Moreover, it remembers the first letter $a$ of $w$ in its finite control and also ensures that $c_0$ is also encoded in the
    Boolean variables using the encoding described above. 
    \item Once an initial configuration is pushed in the stack, non-deterministically, a letter of the form $(t,1)$ where $t = (q_0,a) \rightarrow \{(q_1,a_1,d_1),\dots,(q_k,a_k,d_k)\}$ is pushed. The $t$ is chosen non-deterministically. It then proceeds to the forward mode.
    \item In the  \textbf{forward mode},
    suppose the top of the stack is of the form $(t,i)$ for $t = (q,a) \rightarrow \{(q_1,a_1,d_1),\dots,(q_k,a_k,d_k)\}$ and the contents of the Boolean variables denotes a configuration $c$. Then, from the Boolean variables, $\Ac$ gets the state and the three letters around the head of the configuration $c$. 
    Using these pieces of information, it updates the values
    of the Boolean variables encoding the state and the three letters
    around the head according to the transition $(q,a)\rightarrow (q_i,a_i,d_i)$. 
    Having done this, the Boolean variables now encode a new configuration $c'$.
    $\Ac$ then pushes this configuration $c'$ to the stack.
    After pushing a configuration, the machine then pushes a letter of the form $(t,1)$ where $t$ is chosen non-deterministically, out of the set of transitions that are possible from the configuration $c'$ stored in the Boolean variables.
    \item If a configuration pushed contains a final state, then $\Ac$ goes into \textbf{reverse mode} defined below where the following happens:
    \begin{itemize}
        \item It pops elements from the stack until it reaches a letter of the form $(t,i)$ for $t = (q,a) \rightarrow \{(q_1,a_1,d_1),\dots,(q_k,a_k,d_k)\}$ or the bottom of the stack symbol.
        \item If the bottom of the stack is reached, $\Ac$ has completed its DFS traversal of the accepting tree and reaches an accept state.
        \item If $i=k$, then it pops $(t,i)$ and the configuration below it and continues to be in the reverse mode.
        \item If $i<k$, then it pops $(t,i)$, remembers it in its finite control, pops the next $n+1$ symbols and stores the corresponding configuration that it pops in the Boolean variables. Later it then pushes the same configuration onto the stack and then pushes $(t,i+1)$ onto the top of the stack and proceeds into forward mode. 
    \end{itemize}
\end{itemize}
The above NEPS simulates a run-tree of an alternating linearly bounded Turing machine and can be encoded in size that is polynomial in $|\Mc|$ and $|w|$, showing that emptiness of NEPS is $\EXPTIME$-hard.
\end{proof}

%% file: Appendix2-lower-bound-reduction.tex
\ThmSyncHardness*

\begin{proof}
    We now present the reduction from the emptiness problem for AEPS to {\scshape Special-Sync} in more detail.
    
    Let $\Ac = (Q,V,\Gamma,\Delta,\init,\fin)$ be an AEPS.
    Without loss of generality, we shall assume that 
    if $(q,A,G) \hookrightarrow \{(q_1,\gamma_1,C_1),
    (q_2,\gamma_2,C_2),\dots,(q_k,\gamma_k,C_k)\} \in \Delta$, then $\gamma_i \neq \gamma_j$
    for $i \neq j$. 
    This is because, if it happens that (say) $\gamma_1 = \gamma_2$, then
    we introduce a new stack symbol $\#$, a new state $q'$ and then replace this transition
    with $(q,A,G) \hookrightarrow \{(q',\#\gamma_1,C_1),(q_2,\gamma_2,C_2),\dots,(q_k,\gamma_k,C_k)\},$ 
    $(q',\#,\emptyset) \hookrightarrow \{(q_1,\epsilon,\emptyset)\}$.
    Similarly we can introduce additional states if equality holds for other indices as well.
    Having made this assumption, the desired reduction is described below.

    We now construct a pushdown automaton $\Pc$ as follows:
    The stack alphabet of $\Pc$ will be $\Gamma$.
    For each transition $t \in \Delta$, $\Pc$ will have an input letter $\inputt(t)$.
    $\Pc$ will also have another input letter $\End$.
    The state space of $\Pc$ will be the set $Q \cup (V \times \{0,1\}) \cup \{q_\acc,q_\rej\}$, where $q_\acc$ and $q_\rej$ are two new states, which 
    on reading any input letter, will leave the stack untouched and simply stay at $q_\acc$ and $q_\rej$ respectively. 
    
    Now we describe the transitions of $\Pc$.
    Let $t = (q,A,G) \hookrightarrow \{(q_1,\gamma_1,C_1),\dots,(q_k,\gamma_k,C_k)\}$ be
    a transition of $\Ac$.  
    Let $p \in Q$. Upon reading $\inputt(t)$, if $p \neq q$ then $p$ immediately 
    moves to the $q_\rej$ state. Further, even state $q$ moves to the $q_\rej$ state
    if the top of the stack is not $A$. However, if the top of the 
    stack is $A$, then $q$ pops $A$ and \emph{non-deterministically} pushes 
    any one of $\gamma_1,\dots,\gamma_k$ onto the stack and if it pushed $\gamma_i$, then $q$
    moves to the state $q_i$.
    
    Let $(v,b) \in V \times \{0,1\}$. Upon reading $\inputt(t)$, if the test $v ?= (1-b)$ appears in the guard $G$, then $(v,b)$ immediately moves to the $q_{\rej}$ state. (Notice
    that this is a purely syntactical condition on $\Ac$). 
    Further, if the top of the stack is not $A$, then once again $(v,b)$ moves to $q_{\rej}$.
    If these two cases do not hold, then $(v,b)$ pops $A$ and \emph{non-deterministically} picks an $i\in \{1,\dots k\}$ and pushes $\gamma_i$ onto the stack. Having pushed $\gamma_i$, if $C_i$ does not update the variable $v$,  it stays in state $(v,b)$; otherwise if $C_i$ has a command $v \mapsto b'$, it moves to $(v,b')$.
    
    Finally, upon reading $\End$, the states in $\{\fin\} \cup \{(v,0) : v \in V\}$ move
    to the $q_\acc$ state and all the other states in $Q \cup (V \times \{0,1\})$
    move to the $q_\rej$ state. This ends our construction of $\Pc$.

    Given an assignment $F : V \to \{0,1\}$ of the Boolean variables $V$, and a state $q$ of $\Ac$, we use 
    the notation $\set{q,F}$ to denote the subset $\{q\} \cup \{(v,F(v)) : v \in V\}$ of states of $\Pc$. We now analyse some basic properties of the constructed automaton 
    $\Pc$.
    \begin{itemize}
        \item By construction of $\Pc$, it is easy to see that,
        \begin{quote}
            \emph{Fact A: } Suppose $t$ is a transition of $\Ac$ which is not enabled at the configuration $(q,A\gamma,F)$. Then, upon reading $\inputt(t)$, there is at least one possible successor $(S,\eta)$ of the pseudo-configuration $(\set{q,F},A\gamma)$ such that $q_\rej \in S$.
        \end{quote}
        
        Indeed, suppose $t = (p,B,G) \hookrightarrow \{(p_1,\gamma_1,C_1),\dots,(p_k,\gamma_k,C_k)\}$
        is a transition of $\Ac$ which is not enabled at $(q,A\gamma,F)$.
        Either $q \neq p$, in which case the state $q$ moves to $q_\rej$ in $\Pc$;
        Or $A \neq B$, in which case all the states in $\set{q,F}$ move to $q_{\rej}$ in $\Pc$;
        Or for some variable $v$, the value $F(v)$ does not satisfy some guard in $G$,
        which can happen iff the test $v ?= 1-F(v)$ appears in $G$, in which case
        the state $(v,F(v)))$ moves to $q_{\rej}$ in $\Pc$. This proves Fact A.
        
        \item  Now, recall that if $(q,A,G) \hookrightarrow 
        \{(q_1,\gamma_1,C_1),\dots,(q_k,\gamma_k,C_k)\}$
        is a transition in $\Ac$, then $\gamma_i \neq \gamma_j$ for any $i \neq j$,
        With this in mind, the following fact is rather immediate to see
        \begin{quote}
            \emph{Fact B: } Suppose the configuration $(q,A\gamma,F)$ forks into the configurations $(q_1,\gamma_1\gamma,F_1),\dots,\\(q_k,\gamma_k\gamma,F_k)$ using the transition $t$ in the
            AEPS $\Ac$. 
            Then, the possible successors from the pseudo-configuration $(\set{q,F},A\gamma)$ upon reading $\inputt(t)$ in  the PDA $\Pc$
            are $(\set{q_1,F_1},\gamma_1\gamma),\dots,(\set{q_k,F_k},\gamma_k\gamma)$.
        \end{quote}
    \end{itemize}

    

    Using these 2 facts, we now claim that:
    \begin{quote}
        There exists an accepting run from a configuration $(q,\eta,H)$ in $\Ac$ iff there exists a super-synchroniser between $(\set{q,H},\eta)$ and $q_{\acc}$ in $\Pc$.
    \end{quote}

    ($\Rightarrow$ ) Suppose there is an accepting run from a configuration
    $(q,\eta,H)$ in $\Ac$. We prove the claim by induction on the size
    of the accepting run. For the base case of 1, it must be the
    case that $(q,\eta,H) = (\fin,\bot,\zero)$. In this case,
    by inputting the letter $\End$, it is clear that 
    there is a super-synchroniser between $(\set{\fin,\zero},\bot)$
    and $q_{\acc}$ in $\Pc$.
    
    Suppose we have an accepting run $T$ of size $m+1$ from the 
    configuration $(q,\eta,H)$ in $\Ac$.
    The root of $T$ is labelled by $(q,\eta,H)$. 
    Suppose its children are labelled by $(q_1,\gamma_1,F_1),\dots,(q_k,\gamma_k,F_k)$.
    By induction hypothesis, for each $1 \le i \le k$, we have
    a super-synchroniser between $(\set{q_i,F_i},\gamma_i)$ and $q_{\acc}$ in $\Pc$.
    By Fact B, it follows that we then have a super-synchroniser 
    between $(\set{q,H},\eta)$ and $q_{\acc}$.\\
    
    ($\Leftarrow$) Suppose there exists a super-synchroniser (say $T$)
    between $(\set{q,H},\eta)$ and $q_\acc$ in $\Pc$.
    Without loss of generality, we can assume that only the leaves of $T$ are labelled by $(\{q_\acc\},\bot)$. Let $(S_n,\gamma_n)$ be the label of each node $n$ in $T$. We now proceed to make some observations.
    \begin{itemize}
        \item Since there are no outgoing transitions out of $q_\rej$ and since $S_n = \{q_\acc\}$ for every leaf node $n$, it follows that
        \begin{quote}
            \emph{Fact C: } For every node $n$, $q_\rej \notin S_n$.
        \end{quote}
        \item By induction on the structure of the tree $T$, we prove that
        \begin{quote}
            \emph{Fact D: } If $n$ is a non-leaf node, then $S_n = \set{p_n,F_n}$ for 
            some $p_n \in Q$ and some $F_n : V \to \{0,1\}$.
            Further, if the outgoing edge from some node $n$ is labelled by
            $\End$, then $n$ is the parent of a leaf with $(S_n,\gamma_n) = (\set{\fin,\zero},\bot)$.
        \end{quote}
        
        If $n$ is the root node, then clearly $S_n = \set{q,H}$ and so satisfies the claim.
        Suppose for some non-leaf node $n$, $S_n = \set{p_n,F_n}$. 
        Suppose the outgoing edges from $n$ are labelled by some letter $\inputt(t)$.
        If $t$ is not enabled at the configuration $(p_n,F_n,\gamma_n)$ in $\Ac$,
        by Fact A, there is at least one child $n'$ of $n$ with $q_\rej \in S_{n'}$,
        which contradicts Fact C.
        Hence, $t$ must be enabled at $(p_n,F_n,\gamma_n)$ in $\Ac$.
        By Fact B, it is then clear that for all children
        $n'$ of $n$, $S_{n'}$ is also of the form $\set{p_{n'},F_{n'}}$
        for some $p_{n'} \in Q$ and $F_{n'} : V \to \{0,1\}$.
        
        Hence, the only remaining case is when the outgoing
        edges of $n$ are labelled by $\End$.
        In this case, there is only one child of $n$ (say $n'$).
        If $p_n \neq \fin$ or if $F_n \neq \zero$, then $q_\rej \in S_{n'}$ which contradicts Fact C. Hence, $S_n = \set{\fin,\zero}$ and so $S_{n'} = \{q_\acc\}$. Notice
        that if $\gamma_n$ is not $\bot$, then $\gamma_{n'}$ is also not $\bot$.
        No transition of $q_\acc$ pops the stack and there is no outgoing transition
        from $q_\acc$, and so it would follow that no leaf in the subtree of $n'$
        is labelled by $(q_\acc,\bot)$, which is a contradiction.
        Hence, $\gamma_n = \bot$ and so $\gamma_{n'} = \bot$. Since $(S_{n'},\gamma_{n'}) = (q_\acc,\bot)$, $n'$ is a leaf. Hence, Fact D is true.
                
    \end{itemize}

    By Facts A, B, C and D, it then follows that if we take $T$, remove all its leaves
    and change the label $(\set{p_n,F_n},\gamma_n)$ of each node $n$
    to the label $(p_n,\gamma_n,F_n)$, we get an accepting run of $\Ac$.\\

    Hence there is an accepting run of $\Ac$ iff there is 
    a super-synchroniser between $(\set{\init,\zero},\bot)$ and $q_{\acc}$ in $\Pc$.
    Notice that $\Pc$ is deterministic if $\Ac$ is non-deterministic.  
    Hence, by Lemmas~\ref{lemma:ABP-2EXP} and \ref{lemma:NBP-EXP}, we get the 
    required claims.
\end{proof}

%% file: Appendix3_Homing.tex
Intuitively, in the homing problem, there is an observer who has no knowledge of the
current state of the PDA. The problem then asks if there is a strategy for the observer, to input letters adaptively and narrow down the possible set of states the PDA is in to exactly one state.

We first define  the notion of a homing word from a pseudo-configuration.
Let $\Pc = (Q,\Sigma,\Gamma, \delta)$ be a PDA with $I \subseteq Q$ and $\gamma \in \Gamma^*$ We say that the pseudo-configuration $(I,\gamma)$ admits a \emph{homing word} if there is a labeled tree $T$ satisfying the following conditions
\begin{itemize}
    \item All the edges are labelled by some input letter $a \in \Sigma$ such that, for every vertex $v$, all its outgoing edges have the same label. 
    \item The root is labelled by the pseudo-configuration $(I,\gamma)$. 
    \item Suppose $v$ is a vertex which is labelled by the pseudo-configuration $(S,A\eta)$.
    Let $a$ be the unique label of its outgoing edges and let $\Succ(S,A\eta,a)$ be of size $k$.
    Then $v$ has $k$ children, with the $i^{th}$ child labelled by the $i^{th}$ pseudo-configuration in $\Succ(S,A\eta,a)$.
    \item For every leaf, there exists $q\in Q$ and $\eta \in \Gamma^*$ such that its label is $(\{q\},\gamma')$ for some $\gamma'\in \Gamma^*$.
\end{itemize}
Notice that any synchroniser from some pseudo-configuration $(I,\gamma)$ to some state $s$ is also a homing word from $(I,\gamma)$.
For the automata in Figure~\ref{fig:automata}, we see that the tree in Figure~\ref{fig:tree} is a homing word from $(\{1,2,3,4\},\bot)$. In fact, just a subtree of the one in Figure~\ref{fig:tree}, where we prune the tree as soon as a pseudo configuration with one state is reached is also a homing word from $(\{1,2,3,4\},\bot)$.

The homing problem {\scshape Homing} is now defined as follows:
\begin{quote}
    \noindent \emph{Given: } A PDA $\Pc = (Q,\Sigma,\Gamma,\delta)$ and a word $\gamma \in \Gamma^*$\\
    \noindent \emph{Decide: } If there is a \emph{homing word} from $(Q,\gamma)$
\end{quote}

Similarly the subset homing problem {\scshape Subset-Homing} is defined as:
\begin{quote}
    \noindent \emph{Given: } A PDA $\Pc = (Q,\Sigma,\Gamma,\delta)$, a subset $I \subseteq Q$ and a word $\gamma \in \Gamma^*$\\
    \noindent \emph{Decide: } If there is a \emph{homing word} from $(I,\gamma)$
\end{quote}

By using the same reduction as given in Lemma~\ref{lem:equiv-subset-sync}, it follows that
\begin{lemma}\label{lem:equiv-subset-homing}
    {\scshape Homing} and {\scshape Subset-Homing} are polynomial-time equivalent.
\end{lemma}

We now have the following lemma which relates the homing problem to the adaptive
synchronisation problem.
\begin{lemma}
{\scshape Homing} is polynomial time equivalent to {\scshape Ada-Sync}.
\end{lemma}

\begin{proof}

\paragraph*{Reducing an instance of {\scshape Homing} to {\scshape Ada-Sync}}

By Lemmas~\ref{lem:equiv-subset-sync} and~\ref{lem:equiv-super-sync}, it suffices to show that {\scshape Homing} can be reduced to {\scshape Given-Sync}.

Let $\Pc = (Q,\Sigma,\Gamma,\delta)$ and $\gamma \in \Gamma^*$.
We construct a PDA $\Pc' = (Q',\Sigma',\Gamma,\delta')$  as follows:
$Q'$ consists of all the states of $\Pc$, along with two new states $q_{\acc}$ and $q_{rej}$. $\Sigma'$ is taken to be $\Sigma \cup \{a_q: q \in Q\}$. The transition relation $\delta'$ contains all the transitions in $\delta$ and in addition has the following new ones: 
The state $q \in Q$, upon reading $a_q$ moves to $q_{\acc}$ and upon
reading $a_p$ for some $p \neq q$ moves to $q_{\rej}$. 
We now claim that 
\begin{quote}
    There is a homing word in $\Pc$ from $(Q,\gamma)$ iff there is a synchroniser between $(Q,\gamma)$ and $q_{\acc}$ in $\Pc'$.
\end{quote}

($\Rightarrow$) : Let $T$ be a homing word in $\Pc$ from $(Q,\gamma)$. For each leaf $v$, do the following: Suppose $(\{q\},\eta_v)$ is the label of the leaf $v$. Add an outgoing edge from $v$ 
labelled by $a_q$ and label the child of this edge by $(q_{\acc},\eta_v)$. It is
clear that the modified tree is a synchroniser between $(Q,\gamma)$ and $q_{\acc}$ in $\Pc'$.\\

($\Leftarrow$) : Let $T$ be a synchroniser between $(Q,\gamma)$ and $q_{\acc}$ in $\Pc'$.
We can assume that $T$ is a minimal such synchroniser.
For each vertex $v$, let $(S_v,\eta_v)$ be the label of $v$ in $T$.
Since there are no outgoing transitions from $q_{\rej}$ and since
$S_u = \{q_\acc\}$ for every leaf $u$, it follows
that $q_{\rej} \notin S_v$ for any vertex $v$. We now claim that,

\begin{quote}
    For every non-leaf $v$, $S_v \subseteq Q$. Further, if an outgoing edge
    of $v$ is labelled by $a_q$ for some $q$, then $v$ is the parent of a leaf
    with $S_v = \{q\}$.
\end{quote}

Clearly $S_v \subseteq Q$ when $v$ is the root. Suppose for some non-leaf $v$, $S_v \subseteq Q$. If the outgoing edges from $v$ are labelled by some letter
from $\Sigma$, then it is clear that for all children $v'$ of $v$ we have $S_{v'} \subseteq Q$.
Suppose the outgoing edge from $v$ is labelled by $a_q$ for some $q \in Q$.
Then, there is only one child of $v$ (say $v'$). If $S_v \neq \{q\}$, then $q_{\rej} \in S_{v'}$,
which leads to a contradiction. Hence, $S_v = \{q\}$ and $S_{v'} = \{q_{\acc}\}$.
By minimality of $T$, $v'$ is a leaf.

Hence, it follows that if we remove the leaves of $T$, we get a homing word in $\Pc$ from $(Q,\gamma)$.\\

Hence, there is a homing word in $\Pc$ from $(Q,\gamma)$ iff there is a synchroniser between $(Q,\gamma)$ and $q_{\acc}$ in $\Pc'$.
Therefore, we get  that {\scshape Homing} is reducible to {\scshape Given-Sync}.

\paragraph*{Reducing an instance of {\scshape Ada-Sync} to {\scshape Homing}}

By Lemmas~\ref{lem:equiv-subset-sync},~\ref{lem:equiv-super-sync} and~\ref{lem:equiv-subset-homing}, it suffices
to show that {\scshape Given-Sync} can be reduced
to {\scshape Subset-Homing}. 

Let $\Pc= (Q,\Sigma,\Gamma,\delta)$ be a PDA with $\gamma\in\Gamma^*$, $I\subseteq Q$ and $s\in Q$.
We now construct $\Pc' = (Q',\Sigma',\Gamma,\delta')$ as follows:

The states of $\Pc'$
will be $Q' = (Q \times \{0,1\}) \cup \{q_\acc,(q_\rej,0),(q_\rej,1)\}$.
The input alphabet of $\Pc'$ will be $\Sigma \cup \{\End\}$.
Upon reading any letter from $\Sigma$, the transitions on a state $(q,b) \in Q \times \{0,1\}$ just mimic the transitions of $\delta$ on the first co-ordinate
and leave the second one unchanged. Upon reading $\End$, 
the states $(s,0)$, $(s,1)$ and $q_\acc$ move to $q_\acc$ and 
all the other states move to their corresponding copy of $q_\rej$.

We now claim that
\begin{quote}
    There is a synchroniser from $(I,\gamma)$ to $s$ in $\Pc$ iff there is a homing
    word from  $(I \times \{0,1\},\gamma)$ in $\Pc'$.
\end{quote}

$(\Rightarrow)$: Suppose $T$ is a synchroniser from $(I,\gamma)$ to $s$ in $\Pc$.
Let $(S_v,\eta_v)$ be the label of each vertex $v$ in $T$.
Modify $T$ as follows: For each vertex $v$, replace $S_v$ with $S_v \times \{0,1\}$.
Further, to each leaf $v$ of $T$, add an outgoing edge labelled with $\End$ with the child of this edge being labelled by the pseudo-configuration $(\{q_{\acc}\},\eta_v)$. It is now easy to see that this modified tree $T'$ is a homing word from $(I \times \{0,1\},\gamma)$ in $\Pc'$.\\

$(\Leftarrow)$ : Suppose $\Pc'$ has a homing word from $(I \times \{0,1\}, \gamma)$. Let $T$ be a minimal such homing word. Let $(S_v,\eta_v)$ be the label of each vertex $v$ in $T$.
We claim the following:
\begin{quote}
    For any non-leaf vertex $v$, there exists $Q_v \subseteq Q$ such that $S_v = Q_v \times \{0,1\}$.
    Further, if the outgoing edge from some vertex $v$ is labelled by $\End$, then $v$
    is the parent of a leaf with $S_v = \{s\} \times \{0,1\}$.
\end{quote}

Let us prove this by induction on the structure of $T$. Clearly for the root vertex $v$,
we have $Q_v = I$. Suppose for some non-leaf vertex there exists $Q_v \subseteq Q$ with
$S_v = Q_v \times \{0,1\}$. If the outgoing edges from $v$ are labelled by some letter from $\Sigma$,
then it is easy to see that for every child $v'$ of $v$, there exists $Q_{v'} \subseteq Q$
with $S_{v'} = Q_{v'} \times \{0,1\}$. Suppose the outgoing edge from $v$ is labelled
by $\End$. Hence, there is only one child of $v$ (say $v'$).
If $Q_v \neq \{s\}$, it follows that both $(q_{\rej},0)$ and $(q_{\rej},1)$ belong to 
$S_{v'}$. Notice that there are no outgoing transitions from both of these states.
Hence, for every vertex $v''$ on the subtree rooted at $v'$, we would have $(q_\rej,0) \in S_{v''}$ and
$(q_\rej,1) \in S_{v''}$, which would be a contradiction. Hence,
$Q_v = \{s\}$ and so $S_{v'} = \{q_{\acc}\}$. By minimality of the tree $T$, it follows that $v'$
is a leaf. 

It then follows that if we remove the leaves of $T$ and project the labels of each vertex $v$ to the subset $Q_v$, we will get a synchronising word from $(I,\gamma)$ to $s$ in $\Pc$.

Hence, there is a synchroniser from $(I,\gamma)$ to $s$ in $\Pc$ iff there is a homing
word from  $(I \times \{0,1\},\gamma)$ in $\Pc'$. Therefore, we get that {\scshape Given-Sync} is reducible to {\scshape Subset-Homing}.
\end{proof}